\documentclass[journal]{IEEEtran}

\usepackage{flushend}
\usepackage[linewidth=1pt, skipabove=10pt, skipbelow=10pt]{mdframed}
\usepackage{amsmath}
\usepackage{bm}
\usepackage{graphicx}
\usepackage[ruled,linesnumbered]{algorithm2e}
\usepackage{algpseudocode}
\usepackage{amsfonts,amssymb}
\usepackage{mathrsfs}
\usepackage{multirow}
\usepackage{fmtcount}
\usepackage{url}
\usepackage{setspace}
\usepackage{blindtext}
\usepackage{xcolor,colortbl}
\usepackage{fancyhdr}
\usepackage{booktabs}
\usepackage{dblfloatfix}
\usepackage{anyfontsize}
\usepackage{enumerate}
\usepackage{ntheorem}
\usepackage{pifont}
\usepackage{hhline}

\theoremheaderfont{\bfseries\itshape}
\theorembodyfont{\upshape}
\theoremseparator{:}
\newtheorem{theorem}{Theorem}
\newtheorem{definition}{Definition}
\newtheorem*{proof}{Proof}

\begin{document}

\DeclareRobustCommand\comment[1]{}
\setlength{\textfloatsep}{5pt plus 1.0pt minus 2.0pt}
\setlength{\floatsep}{5pt plus 1.0pt minus 2.0pt}

\title{\fontsize{23.5}{30}\selectfont Network-Wide Task Offloading With LEO Satellites: A Computation and Transmission Fusion Approach}

\author{Jiaqi Cao,\ 
Shengli Zhang,~\IEEEmembership{Senior Member,~IEEE,}\\
Qingxia Chen,\ 
Houtian Wang,\ 
Mingzhe Wang,\ 
Naijin Liu\ 

\thanks{Naijin Liu is the corresponding author.}  
\thanks{Jiaqi Cao and Shengli Zhang are with Shenzhen University, Shenzhen, 518052, P.R. China (e-mail: jiaqicao@szu.edu.cn; zsl@szu.edu.cn).}
\thanks{Qingxia Chen, Houtian Wang and Naijin Liu are with Qian Xuesen Laboratory of Space Technology, China Academy of Space Technology (e-mail: chenqingxia@qxslab.cn; wanghoutian@qxslab.cn; liunaijin@qxslab.cn).}
\thanks{Mingzhe Wang is with Tsinghua University (e-mail: wmzhere@gmail.com).}
}
\maketitle
\thispagestyle{fancy}

\begin{abstract}
Computing tasks are ubiquitous in space missions. Conventionally, these tasks are offloaded to ground servers for computation, where the transmission of raw data on satellite-to-ground links severely constrains the performance. To overcome this limitation, recent works offload tasks to visible low-earth-orbit (LEO) satellites. However, this offloading scheme is difficult to achieve good performance in actual networks with uneven loads because visible satellites over hotspots tend to be overloaded. Therefore, it is urgent to extend the offloading targets to the entire network.

To address the network-wide offloading problem, we propose a metagraph-based computation and transmission fusion offloading scheme for multi-tier networks. Specifically, virtual edges, between the original network and its duplicate, are generated to fuse computation and transmission in single-tier networks. Then, a metagraph is employed to integrate the fusion in multi-tier networks. After assigning appropriate edge weights to the metagraph, the network-wide offloading problem can be solved by searching the shortest path.

In addition, we apply the proposed scheme to solve the spatial computation offloading problem in a real multi-tier network. The superiority of the proposed scheme over two benchmark schemes are proved theoretically and empirically. 

Simulation results show that the proposed scheme decreases the weighted average delay by up to 87.51\% and 18.70\% compared with the ground offloading scheme and the visible offloading scheme, respectively.

\end{abstract}

\IEEEpeerreviewmaketitle%

\section{Introduction}\label{Introduction}
\IEEEPARstart{C}{omputing} tasks are becoming increasingly complex in space
missions. Take remote sensing (RS) for example, the advancements in image processing and target recognition techniques, especially the application of machine learning techniques~\cite{7486259}, have led to a rapid increase in computational requirements~\cite{5678816,917889}. In addition, sensing technology improvements, such as hyperspectral image (HSI)~\cite{8113122}, enable increased data precision at the cost of the huge data volume of remote sensing images.
The growth in both computational requirements and data volume of computing tasks, especially latency-sensitive tasks, create significant challenges to conventional task offloading schemes.

The \textit{ground-server-based computation offloading} scheme (hereinafter referred to the ground offloading scheme)~\cite{9210567,9651919,7805169,9013462} is one of the most popular conventional computing task offloading schemes. In this scheme, low-earth-orbit (LEO) satellites act as bent pipes to route raw data of RS tasks to ground servers for computation. While the ground servers are powerful in computing capability, transmitting large amounts of raw data across the network can be a significant bottleneck: perturbed by the atmosphere frequently~\cite{alonso2004performance}, the transmission rate of satellite-to-ground links (SGLs) can be as low as 20 Mbps on state-of-the-art satellites~\cite{yost2021state}. To overcome the limitations of the ground offloading scheme, researchers have set their sights on onboard computing.

The \textit{visible-LEO-satellite-based computation offloading} scheme (hereinafter referred to the visible offloading scheme) \cite{ZHAO202294,9383778,Kothandhapani2020} is a typical task offloading scheme featuring onboard computing. In this scheme, tasks are first transmitted from the source to its visible LEO satellites directly; next, the LEO satellites process the tasks and route the computing results to the ground.
This scheme can usually achieve impressive performance compared with ground offloading for the following reasons.
\comment{Citations are needed for the reasons.}
\ding{172} LEO satellites are geographically closer to RS satellites than ground servers, removing the constraints of data rates for SGLs.
\ding{173} Onboard computing greatly reduces the amount of data that needs to be transmitted across the network, which lowers the transmission delay in turn.

\begin{figure}[b]
	\centering
	\includegraphics[width=0.48\textwidth]{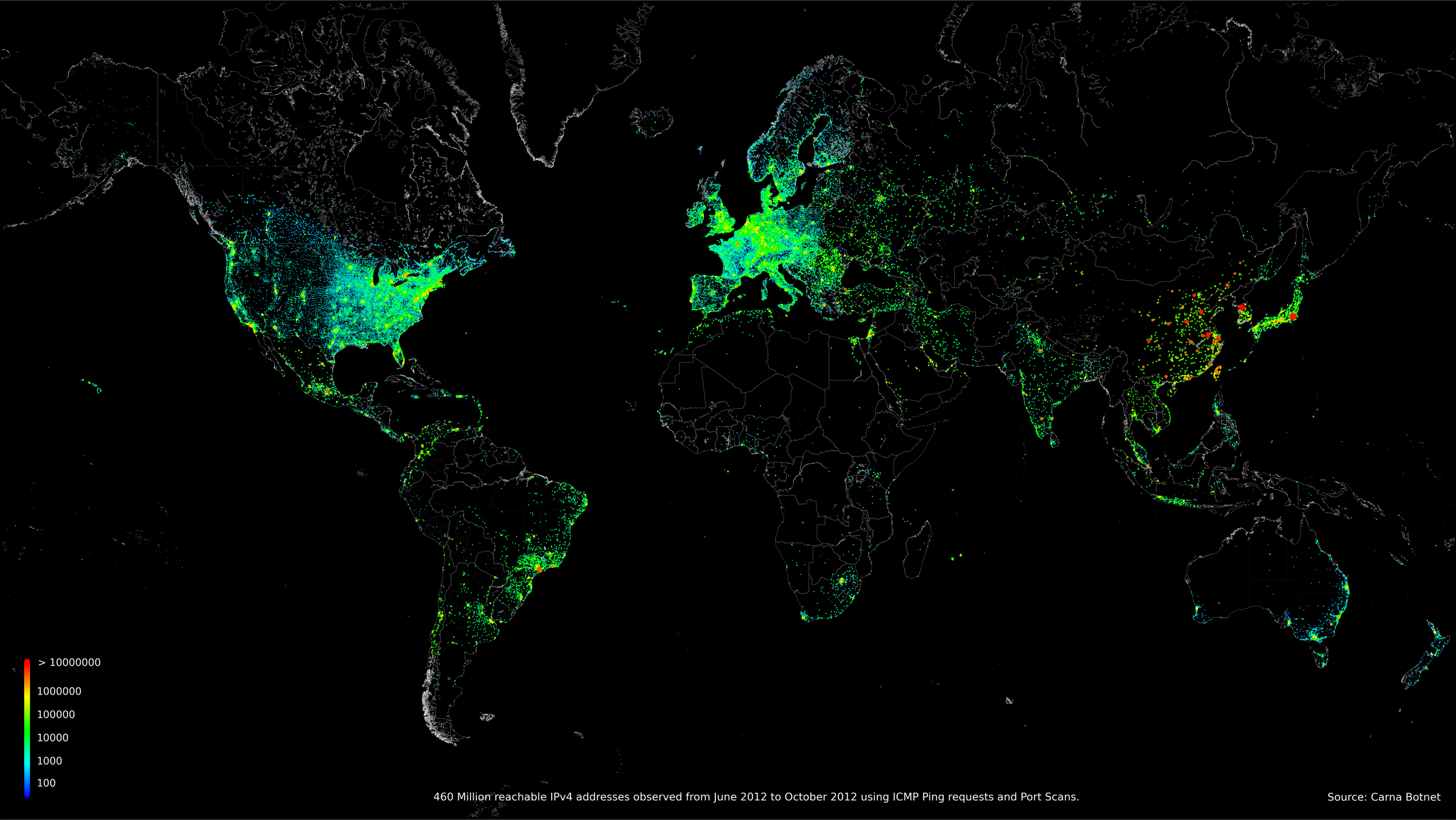}\\
	\begin{center}
		\caption{Global distribution of Internet users according to IPv4 address usage~\cite{InternetConsensus2012}.}\label{WorldMap}
	\end{center}
\vspace{-1.0em}
\end{figure}

However, in reality, the visible offloading scheme is unlikely to achieve the desired performance: computing resources are scarce over hot-spot areas because the visible satellites are usually overloaded.
The root cause of this phenomenon is the uneven distribution of tasks across the world.
As Fig.~\ref{WorldMap} shows, most human activities are located in a few hotspots such as cities and ports. Although lots of computing tasks are generated there in each time slot, only a limited number of satellites are visible and available for offloading during each time slot. In contrast, the major part of computing resources located over non-hotspot areas remain unused.

Network-wide computation offloading could be a promising approach to address the hot-spot challenge.
First, the uneven network load can be balanced by forward tasks to invisible satellites for computation. In this way, the offloading targets are extended to the entire network. It worth noting that the network-wide offloading is complementary to the other two offloading schemes introduced above: when computing resources are abundant, visible satellites can be chosen to reduce delay; when visible satellites are overloaded, tasks can be forward to invisible satellites; for extremely overloaded cases, the ground servers can be chosen as the last resort.

The key challenge in network-wide offloading is the joint optimization of the transmission process and the computation process. The overall delay\footnote{The overall delay refers to the moment from a task is generated to the moment when the destination obtains the computing results of that task.} is determined by three stages, i.e., transmission (source \(\to\) offloading target), computation, and transmission (offloading target \(\to\) destination). While \textit{both} transmission and computation processes affect the overall delay, existing works only address the optimization problem of \textit{either} process and suffer from local minimums.
For example, some offloading approaches, including the ground offloading scheme, offload tasks to specific computing nodes and focus on finding the optimal path from the source to those specific computing nodes. However, these approaches are unable to evaluate and select the most appropriate offloading target among numerous potential offloading targets with computing capability.
Other offloading approaches, including the visible offloading scheme, offload tasks to computing nodes within one-hop range and focus on finding the optimal task distribution among multiple computing nodes. However, the available computing resources in these approaches are very limited.
Despite the extensive studies on \textit{either} problem, a fusion offloading scheme covering \textit{both} of them is still needed.

\subsection{Main Contributions}
In this paper, we propose a metagraph-based computation and transmission fusion task offloading scheme.
The key idea is to integrate multi-tier networks according to offloading constraints, such that well-studied path-finding algorithms can be repurposed for offloading.
To this end, the virtual edge and metagraph are introduced for graph abstraction.

We use virtual edges to abstract multi-process offloading states.
To unify the representation of computation and transmission processes, we first duplicate the original network and generate virtual edges by connecting the corresponding nodes in the original network and the network duplicate.
Each node in the original network and its duplicate represent a \textit{non-computed state} and an \textit{already-computed state}, respectively.
In other words, each virtual edge can be considered as a computational state migration, and thus can be used to represent the computation process of the corresponding node. 
Therefore, we can assign weights (e.g., computation delays) to virtual edges according to the computation process.
This resembles classic path-finding algorithms, which assign weights to actual edges according to the transmission processes. The similar designs of actual edge and virtual edge unify the representation of both processes, allowing joint optimization towards the same goal.

We use metagraphs to abstract multi-tier networks.
To achieve joint optimization over multi-tier networks, we construct a metagraph to connect the graphs, where each graph can be a single-tier network. These graphs are connected with actual edges representing the inter-tier transmission process. 
Therefore, a \textit{multi-process} fused \textit{multi-tier} network is represented with \textit{one} metagraph, where tasks can be offloaded to any computing nodes across the whole network. In this way, the minimum overall delay of transmission and computation can be obtained by searching the optimal path in the metagraph after assigning an appropriate weight to each edge.
%
The main contributions are summarized as follows.
\begin{enumerate}
  	\item We observe that the performance (e.g., delay) of existing LEO-satellite-based computation offloading schemes is crippled by the unbalanced network load in real-world networks.
    To overcome this limitation, we present that the offloading targets need to be extended to the entire network. The network-wide computation offloading problem is transformed into a joint optimization problem over the transmission and computation processes.
    Different from existing studies that focus on the joint allocation of computation and transmission resources, the joint optimization problem discussed in this paper also includes the transmission path optimization.
	\item We propose a virtual-edge-based computation and transmission fusion method to solve the joint optimization problem for single-tier networks.
	Similar to actual edges (representing transmission processes in existing path-finding algorithms), virtual edges represent computation processes.
	Consequently, the transmission and computation process are uniformly represented with edges, which enables the multi-process fusion in single-tier networks.
	\item We propose a metagraph-based method to achieve network-wide offloading inside multi-tier networks. The metagraph connects multiple single-tier networks based on inter-tier visibility.
	With assigning appropriate weights to the edges of the metagraph, multiple single-tier networks are reorganized into a large fusion network.
	Consequently, the joint optimization of both computation and transmission processes is converted to the well-studied path-finding problem over the fusion network.
	\item We show the superiority of our fusion offloading scheme theoretically and empirically.
	From the theoretical perspective, we represent the existing ground offloading scheme and orbital edge computing scheme as metagraphs, and then prove the superiority of our scheme using graph theory.
	From the empirical perspective, we simulate the performance of offloading schemes with typical RS settings.
	Results show that the proposed scheme can reduce the weighted average delay by up to 87.51\% and 18.70\% compared with the ground offloading scheme and the visible offloading scheme, respectively.
\end{enumerate}

\begin{table*}[b]
	\centering
	\vspace{-1.0em}
	\begin{spacing}{1}
		\caption{Comparison of Computation Offloading Studies}\label{RelatedWorkTable}
	\end{spacing}
	\arrayrulecolor{black}
	\small
	\resizebox{1.0\textwidth}{!}{\arrayrulecolor{black}
\begin{tabular}{|c|c|c|c|c|c|c|c|c|c|} 
	\hline
	\rowcolor[rgb]{0.812,0.835,0.918} {\cellcolor[rgb]{0.812,0.835,0.918}}                           & {\cellcolor[rgb]{0.812,0.835,0.918}}                            & {\cellcolor[rgb]{0.812,0.835,0.918}}                           & {\cellcolor[rgb]{0.812,0.835,0.918}}                                                                                                                      & \multicolumn{3}{c|}{Offloading Target}                                                                                                                                                                                                                             & \multicolumn{3}{c|}{Optimization}                                                                                                                                                                                                                                                                                                                                                                                                                                           \\ 
	\hhline{|>{\arrayrulecolor[rgb]{0.812,0.835,0.918}}---->{\arrayrulecolor{black}}------|}
	\rowcolor[rgb]{0.812,0.835,0.918} {\cellcolor[rgb]{0.812,0.835,0.918}}                           & {\cellcolor[rgb]{0.812,0.835,0.918}}                            & {\cellcolor[rgb]{0.812,0.835,0.918}}                           & {\cellcolor[rgb]{0.812,0.835,0.918}}                                                                                                                      & {\cellcolor[rgb]{0.812,0.835,0.918}}                         & \multicolumn{2}{c|}{Space}                                                                                                                                                                          & {\cellcolor[rgb]{0.812,0.835,0.918}}                                                                                                                   & {\cellcolor[rgb]{0.812,0.835,0.918}}                                                                                                                      & {\cellcolor[rgb]{0.812,0.835,0.918}}                                                                                                                   \\ 
	\hhline{|>{\arrayrulecolor[rgb]{0.812,0.835,0.918}}----->{\arrayrulecolor{black}}-->{\arrayrulecolor[rgb]{0.812,0.835,0.918}}--->{\arrayrulecolor{black}}|}
	\rowcolor[rgb]{0.812,0.835,0.918} \multirow{-3}{*}{{\cellcolor[rgb]{0.812,0.835,0.918}}Category} & \multirow{-3}{*}{{\cellcolor[rgb]{0.812,0.835,0.918}}Reference} & \multirow{-3}{*}{{\cellcolor[rgb]{0.812,0.835,0.918}}Scenario} & \multirow{-3}{*}{{\cellcolor[rgb]{0.812,0.835,0.918}}\begin{tabular}[c]{@{}>{\cellcolor[rgb]{0.812,0.835,0.918}}c@{}}Network\\Architecture\end{tabular}} & \multirow{-2}{*}{{\cellcolor[rgb]{0.812,0.835,0.918}}Ground} & \begin{tabular}[c]{@{}>{\cellcolor[rgb]{0.812,0.835,0.918}}c@{}}Visible\\Satellite\end{tabular} & \begin{tabular}[c]{@{}>{\cellcolor[rgb]{0.812,0.835,0.918}}c@{}}Invisible\\Satellite\end{tabular} & \multirow{-2}{*}{{\cellcolor[rgb]{0.812,0.835,0.918}}\begin{tabular}[c]{@{}>{\cellcolor[rgb]{0.812,0.835,0.918}}c@{}}Computing\\Resource\end{tabular}} & \multirow{-2}{*}{{\cellcolor[rgb]{0.812,0.835,0.918}}\begin{tabular}[c]{@{}>{\cellcolor[rgb]{0.812,0.835,0.918}}c@{}}Transmission\\Resource\end{tabular}} & \multirow{-2}{*}{{\cellcolor[rgb]{0.812,0.835,0.918}}\begin{tabular}[c]{@{}>{\cellcolor[rgb]{0.812,0.835,0.918}}c@{}}Transmission\\Path\end{tabular}}  \\ 
	\hline
	1                                                                                                & \cite{9651919,7805169,9013462}                                  & \multirow{4}{*}{Satellite}                                     & Mesh                                                                                                                                                      & \checkmark                                                   &                                                                                                 &                                                                                                   &                                                                                                                                                        & \checkmark                                                                                                                                                & \checkmark                                                                                                                                             \\ 
	\cline{1-2}\cline{4-10}
	2                                                                                                & \cite{9148261}                                                  &                                                                & \multirow{6}{*}{Tree}                                                                                                                                     &                                                              & \checkmark                                                                                      &                                                                                                   & \checkmark                                                                                                                                             &                                                                                                                                                           &                                                                                                                                                        \\ 
	\cline{1-2}\cline{5-10}
	\multirow{2}{*}{3}                                                                               & \cite{9383778,9515574}                                  &                                                                &                                                                                                                                                           &                                                              & \checkmark                                                                                      &                                                                                                   & \checkmark                                                                                                                                             & \checkmark                                                                                                                                                &                                                                                                                                                        \\ 
	\cline{2-2}\cline{5-10}
	& \cite{9344666,9793590}                                                  &                                                                &                                                                                                                                                           & \checkmark                                                   & \checkmark                                                                                      &                                                                                                   & \checkmark                                                                                                                                             & \checkmark                                                                                                                                                &                                                                                                                                                        \\ 
	\cline{1-3}\cline{5-10}
	\multirow{3}{*}{4}                                                                               & \cite{8387798,8663968,8713498}                                  & Cellular                                                       &                                                                                                                                                           & \checkmark                                                   &                                                                                                 &                                                                                                   & \checkmark                                                                                                                                             & \checkmark                                                                                                                                                &                                                                                                                                                        \\ 
	\cline{2-3}\cline{5-10}
	& \cite{9447257,8877759}                                          & UAV                                                            &                                                                                                                                                           & \checkmark                                                   &                                                                                                 &                                                                                                   & \checkmark                                                                                                                                             & \checkmark                                                                                                                                                &                                                                                                                                                        \\ 
	\cline{2-3}\cline{5-10}
	& \cite{8936886}                                          & IoT                                                            &                                                                                                                                                           & \checkmark                                                   &                                                                                                 &                                                                                                   & \checkmark                                                                                                                                             & \checkmark                                                                                                                                                &                                                                                                                                                        \\ 
	\hline
	\rowcolor[rgb]{0.812,0.835,0.918} \multicolumn{2}{|c|}{ Our work}                                                                                                  & Satellite                                                      & Mesh                                                                                                                                                      & \checkmark                                                   & \checkmark                                                                                      & \checkmark                                                                                        & \checkmark                                                                                                                                             & \checkmark                                                                                                                                                & \checkmark                                                                                                                                             \\
	\hline
\end{tabular}}
\end{table*}

\section{Related Work}\label{RelatedWork}
In this section, existing studies on computation offloading are introduced first. Then we summarize the state-of-the-art of LEO satellite capabilities, including computing capabilities, data transmission rates of inter-satellite links (ISLs) and SGLs.
\subsection{Existing Works on Computation Offloading}
Existing works on computation offloading can be categorized according to network architectures and optimization problems. Table~\ref{RelatedWorkTable} summarizes the differences between prior works and our work.

\textbf{Category~1} in Table~\ref{RelatedWorkTable} lists typical studies addressing the data transmission problem from space to ground servers for computation via LEO satellite networks. These studies are classified as the \textit{ground offloading} scheme in this paper. 
For example, the work in~\cite{9651919} presented an data delivery architecture upon LEO constellations and ground stations to forward data from remote sensing satellites to end users. 
The authors in~\cite{7805169} proposed a collaborative scheme that allowed satellites to offload data among themselves using ISLs before they came into contact with the earth station (ES), such that satellites would carry the right amount of data according to the length of their contact time with the ES. The work in~\cite{9013462} proposed an iterative algorithm to minimize the energy consumption of data transmission while maximizing the throughput to address data offloading in space. 
In these studies, each LEO satellite is allowed to communicate with its neighboring LEO satellites via ISLs, thus forming a \textit{mesh network}. Although the resource allocation and path finding problems in transmission have been optimized, these studies do not address the high delay problem caused by transmitting large amounts of raw data across the network. To overcome this limitation, onboard computing entered the limelight.
 
\textbf{Category~2} in Table~\ref{RelatedWorkTable} lists typical studies that employ onboard computing and investigate the optimal computing task distribution among multiple visible (i.e., communicable) LEO satellites. These studies are classified as the \textit{visible offloading} scheme in this paper. 
For example, authors in~\cite{9148261} introduced space stations to offload computing tasks of LEO satellites and investigated the computing resource allocation problem between LEO satellites and space stations. 
In these studies, LEO satellites are individual from each other. Each LEO satellite only communicates with devices (e.g., ground servers) located within its coverage, which forms a \textit{tree network}. In this network architecture, 
tasks cannot be offloaded to other computing devices in the same network tier. In this condition, the computation offloading strategies in Category~2 may not perform well in actual networks. This is because computing resources in LEO constellations are almost uniformly distributed; whereas, the global task distribution in actual networks is extremely uneven. To further improve the offloading performance, the network-wide computation offloading is urgently needed. 

To achieve network-wide computation offloading, two problems need to be solved simultaneously: the computing node selection problem and the computing task transmission problem. Therefore, the joint optimization of the computing process and transmission process is needed. Category~3 and Category~4 in Table~\ref{RelatedWorkTable} list typical works on the joint optimization of computation and transmission resource allocation.

\textbf{Category~3} in Table~\ref{RelatedWorkTable} enumerates some existing works on computation and transmission joint optimization for satellite networks. In these works, satellites are individual from each other.
The work in~\cite{9383778} and~\cite{9515574} investigated the joint computation assignment and resource allocation problem in multi-tier computing architectures composed of mobile devices, LEO satellites, etc.
Authors in~\cite{9344666} and \cite{9793590} proposed hybrid computation offloading architectures to solve the joint computation and resource allocation problem, where computing tasks could be offloaded to both ground servers and visible LEO satellites. 

\textbf{Category~4} in Table~\ref{RelatedWorkTable} enumerates some existing studies on computation and transmission integration for terrestrial networks.
The work in~\cite{8387798} discussed the joint communication and computation resource allocation in a two-tier device--cloud network where tasks could be processed locally, in the edge cloud or both. 
Authors in~\cite{8663968} and~\cite{8713498} investigated the task offloading problem in fog-enabled cellular networks where radio, caching and computing were jointly optimized.
The work in~\cite{9447257} and~\cite{8877759} proposed joint communication and computation resource scheduling approaches for unmanned aerial vehicle (UAV)-assisted local--edge/local--edge--cloud computing systems, where each UAV worked as an edge computing devices to assist devices within its communicable range.
Authors in~\cite{8936886} developed a cloud-fog-device computing architecture for internet of things (IoT), where the offloading ratio, transmission power and local CPU computation speed were jointly optimized. 

Although the studies mentioned above jointly optimize the allocation of multiple resources, they still fail to achieve network-wide computation offloading. 
This is because the networks in the above studies are \textit{tree networks}, where tasks cannot be forwarded to other computing devices in the same network tier.
To overcome this limitation, LEO satellites in this paper are connected with ISLs, which forms a \textit{mesh network}. 
In this condition, computing tasks can be offloaded to any satellite via routing, which makes it possible to extend the offloading targets to the entire network. However, the network-wide computation offloading for LEO satellites brings a novel challenge: the transmission path needs to be optimized as well. Therefore, we proposed a metagraph-based computation and transmission fusion computation offloading scheme to address this challenge. 

%
\subsection{State-of-the-Art LEO Satellite Capabilities}\label{Background}
The height of LEO satellites usually ranges from 500 km to 2,000 km~\cite{9003618}. Due to the superiority in latency, cost, development cycle, etc., the LEO satellite network is deemed as the most prospective satellite mobile communication system. As the computational requirements and data volume of space missions increase, unprecedented interest and efforts have been devoted to enhancing the computing and transmission capabilities of LEO satellites. The state-of-the-art capabilities of LEO satellites are introduced as follows.
\subsubsection{\textbf{Computing Capability of LEO Satellites}}\label{LEOComputation}
For next-generation science and defense missions, spacecrafts such as LEO satellites must provide advanced processing capability to support a variety of computationally intensive tasks~\cite{geist2019spacecube}. 
The desire for even more onboard processing capacity has led to the development of onboard computing systems. The computing capabilities of some typical highly integrated onboard computing systems are summarized in Table~\ref{onboard Computing Systems}.

\begin{table}[h]
	\centering
	\vspace{-0.5em}
	\begin{spacing}{1}
		\caption{Computing Capabilities of Highly Integrated Onboard Computing Systems (in GFLOPS)}\label{onboard Computing Systems}
	\end{spacing}
	\small
	\resizebox{0.49\textwidth}{!}{%
\begin{tabular}{cccc} 
	\rowcolor[rgb]{0.812,0.835,0.918} Product             & Processor                                                                                                                  & \begin{tabular}[c]{@{}>{\cellcolor[rgb]{0.812,0.835,0.918}}c@{}}Computing\\Capability\end{tabular} & Reference                                     \\
	Xiphos
	
	Q7S                                           & Xilinx Zynq 7020                                                                                                           & 180                                                                                                & \cite{XiphosQ7S}                              \\
	\rowcolor[rgb]{0.914,0.922,0.961} Xiphos Q8S          & Xilinx Ultrascale+                                                                                                         & 1800                                                                                               & \cite{XiphosQ8S}                              \\
	BAE RAD5545                                           & RAD5545                                                                                                                    & 3.7                                                                                                & \cite{BAERAD5545}                             \\
	\rowcolor[rgb]{0.914,0.922,0.961} Innoflight CFC-500  & \begin{tabular}[c]{@{}>{\cellcolor[rgb]{0.914,0.922,0.961}}c@{}}Xilinx Kintex Ultrascale+, \\NVIDIA TK1\end{tabular}       & 1290                                                                                               & \cite{CFC500}                                 \\
	MOOG G-Series Steppe Eagle                            & AMD G-Series compatible                                                                                                    & 75                                                                                                 & \cite{MOOG}                                   \\
	\rowcolor[rgb]{0.914,0.922,0.961} MOOG V-Series Ryzen & AMD V-Series compatible                                                                                                    & 1000                                                                                               & \cite{MOOG}                                   \\
	Unibap iX5-100                                        & \begin{tabular}[c]{@{}c@{}}Microchip SmartFusion2,\\AMD G-Series SOC\end{tabular}                                          & 127                                                                                                & \cite{iX5-100,bruhn2020enabling}              \\
	\rowcolor[rgb]{0.914,0.922,0.961} Unibap iX10-100     & \begin{tabular}[c]{@{}>{\cellcolor[rgb]{0.914,0.922,0.961}}c@{}}Microchip PolarFire, \\AMD V1605b (Ryzen)\end{tabular}     & 3600                                                                                               & \cite{iX10-100,bruhn2020enabling}             \\
	SpaceCube v2.0                                        & Xilinx Virtex 5                                                                                                            & 200                                                                                                & \cite{SpaceCube2.0}                           \\
	\rowcolor[rgb]{0.914,0.922,0.961} SpaceCube v3.0      & \begin{tabular}[c]{@{}>{\cellcolor[rgb]{0.914,0.922,0.961}}c@{}}Xilinx Kintex UltraScale,\\ Xilinx Zynq MPSoC\end{tabular} & 590                                                                                                & \cite{geist2019spacecube,SpaceCube3.0Patent}  \\
\end{tabular}}
\end{table}

It can be concluded form Table~\ref{onboard Computing Systems} that existing onboard computing systems can provide thousands Giga floating-point operations per second (GFLOPS) of computing capability. For example, national aeronautics and space administration (NASA) Goddard Space Flight Center (GSFC) developed SpaceCube v3.0 in 2019~\cite{9546282}. It contains a Xilinx Kintex UltraScale with a Xilinx Zynq MPSoC to provide 10--100x or more performance over other flight single-board computers~\cite{geist2019spacecube}. In specific, the computing capacities of these systems on chips (SoCs) both exceed 100 GFLOPS. 
Although the computing capability of LEO satellites is not yet comparable to that of geosynchronous equatorial orbit (GEO) satellites and ground servers, the prospect and importance of increasing the computing capability of LEO satellites has been recognized and a great deal of research has been invested, indicating a promising future for onboard computing.
%

\subsubsection{\textbf{Transmission Capability of LEO Satellites}}\label{LEODataRate}
ISLs in free space are usually higher in data rate. For instance, as shown by Del Portillo~\cite{del2019technical}, the data rate of optical ISLs can achieve 5 Gbps. Mynaric's laser terminal for LEO constellations is capable of delivering 10 Gbps with a low SWaP unit over a wide range of constellation configurations~\cite{10.1117/12.2545629}. It can operate within densely packed constellations with intra/inter-plane link distances up to 7,800 km.

In contrast, the data rate of SGL cannot keep up with that of ISL due to the perturbation induced by the atmosphere~\cite{alonso2004performance}. The down-link data rate for state-of-the-art satellites ranges from 20 Mbps to 1 Gbps~\cite{yost2021state}. For example, the CubeSat lasercom module by Hyperion Technologies enables a bidirectional space-to-ground communication link between a CubeSat and an optical ground station, with a down-link speed up to 1 Gbps and an up-link data rate of 200 Kbps~\cite{yost2021state}. The limited SGL transmission capability further promotes the application of on-board computing.

\section{System Model}\label{SystemModel}
In this section, the LEO satellite network and its parameter settings adopted in this paper are introduced first. Next, the definition and distribution of tasks are explained. The delay model is introduced in the end.
\subsection{Network Model}\label{Network Model}
Although the proposed fusion offloading scheme is compatible to a wide range of LEO constellations, for demonstration use, we take a Walker Star constellation~\cite{walker1984satellite} with polar orbits (i.e., the orbit inclination angle $ i_0 $ is $ 90^{\circ} $) as an example.
As shown in Fig~\ref{WalkerConstellation} (i), the LEO constellation adopted in this paper consists $ N_o $ circular orbits whose altitude is $ h $ kilometers. These orbit planes are evenly distributed over $ 180^{\circ} $ range, traveling north on one side of the Earth, and south on the other side. They cross each other only over the North and South poles. $ N_s $ LEO satellites, whose maximum computing capability is $ C $, uniformly distributed over each orbit. 
The orbital period is $ T_O = 2\times\pi\times\sqrt{{(R_e+h)}^3/(G\times M_e)} $, where $ R_e $ and $ M_e $ represent the radius and mass of the earth respectively, and $ G $ is the gravitational constant. In addition, LEO satellites connect with each other via ISLs and connect with ground stations via SGLs. The maximum data transmission rate of each ISL and SGL are $ R_{ISL} $ and $ R_{SGL} $, respectively.

\begin{figure*}[t]
	\centering
	\vspace{-0cm}  
	\setlength{\abovecaptionskip}{-0.2cm}
	\setlength{\belowcaptionskip}{-1.0cm}
	\includegraphics[width=0.95\textwidth]{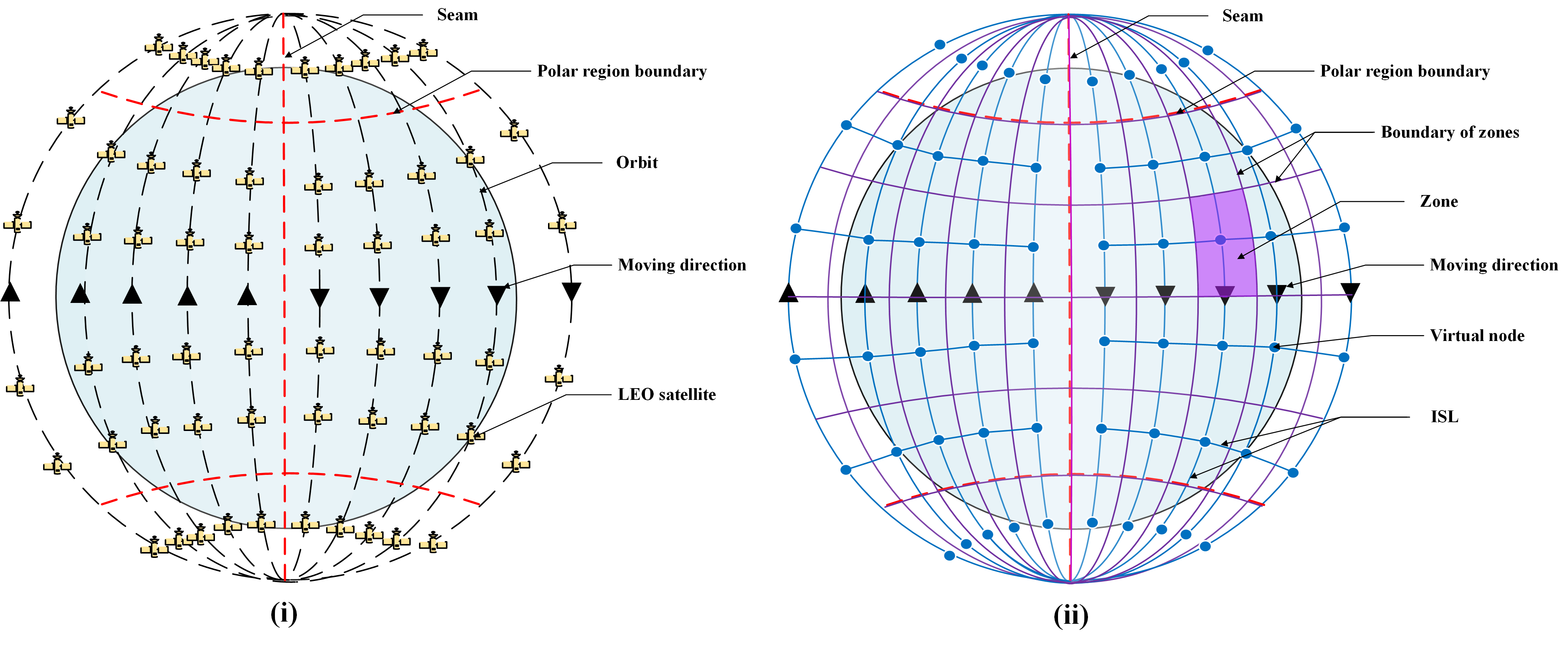}\\
	\begin{center}
		\caption{(i) The Walker constellation. (ii) The virtual-node-based network model of the Walker constellation.}\label{WalkerConstellation}
	\end{center}
	\vspace{-2.0em}
\end{figure*}

As shown in Fig~\ref{WalkerConstellation} (ii), this paper uses virtual nodes (VNs) to model the dynamic LEO satellite network. It divides the LEO satellite network into logical locations (i.e., static zones) which are represented by virtual nodes and associated with the nearest satellites~\cite{2001A}. The association between the static zones and satellites changes due to satellite movements. With this architecture, each change in the satellite association represents a new snapshot~\cite{8332924}. Each snapshot could be considered as a mesh network presenting a static state of the network topology. Each VN has four ISLs with its neighbors where two are intra-plane and two are inter-plane. The ISLs in cross-seam and Polar Regions are switched off due to high dynamic motions.  


\subsection{Traffic Model}\label{TrafficModel}
\subsubsection{\textbf{Task Definition}}
For any task $ \mathcal{T} $, it can be divided into multiple independent subtasks, i.e., $ \mathcal{T}= \{\tau_1,\tau_2,\dots,\tau_l,\dots,\tau_{n}\} \ (l,n\in Z^+,\ l\leq n) $, where $ \tau_l $ is the $ l^{\rm th}$ subtask of $ \mathcal{T}$ and $ n $ is the total number of subtasks in $ \mathcal{T}$. In this paper, subtasks are the smallest unit of transmission and computation, which cannot be further separated. 

Generally, there are five items used to depict task $ \mathcal{T} $, i.e., $ \Lambda_ {\mathcal{T}} = (u,v,n,t,\vartheta) $, where $ u $ is the source node where task $ \mathcal{T} $ is generated, $ v $ is the destination node, $ n $ is the number of subtasks contained in $ \mathcal{T} $, $ t $ is the instance when $ \mathcal{T} $ is generated in seconds and $ \vartheta $ is the delay threshold (i.e., the maximum tolerable delay) of $ \mathcal{T} $ in seconds.

Similarly, any subtask $ \tau_l \in \mathcal{T} $ can be described as $ \Lambda_ {\tau_l} = (\widetilde{C_l},\widetilde{N_l}) $, where $ \widetilde{C_l} $ and $ \widetilde{N_l} $ are the computational requirement of subtask $ \tau_l $ in GFLOPS, the data volume of subtask $ \tau_l $ in gigabytes (GB), respectively.

\subsubsection{\textbf{Task Distribution}}
In this paper, the task arrival of zone $ i $ is assumed to be a Poisson stochastic process with an arrival rate $ \lambda_i $ because such processes have attractive theoretical properties~\cite{267444}. 
Since the transmission of tasks depends on network connections, in this paper we utilize the network connection data from Internet Census 2012 introduced in Section~\ref{Introduction} to estimate the arriving rate of each zone.
First, we divide the Earth into $ 22.5◦^{\circ} \times 22.5^{\circ} $ geographical zones whose set is $ \mathbf{Z} $. 
Next, the open source data of the Internet Census 2012 are utilized to estimate the connection index $ \eta $ of each zone. The results are shown in Fig~\ref{ConnectionIndex}. It is obvious that the distribution of network connections is extremely uneven. 
Finally, the arriving rate $ \lambda_i $ of zone $ i $ is proportional to the corresponding connection index $ \eta_i $ (i.e., $ \lambda_i \propto \eta_i $).

\begin{figure}[h]
	\centering
	\setlength{\abovecaptionskip}{-0.cm}
	\setlength{\belowcaptionskip}{-0.cm}
	\includegraphics[width=0.5\textwidth]{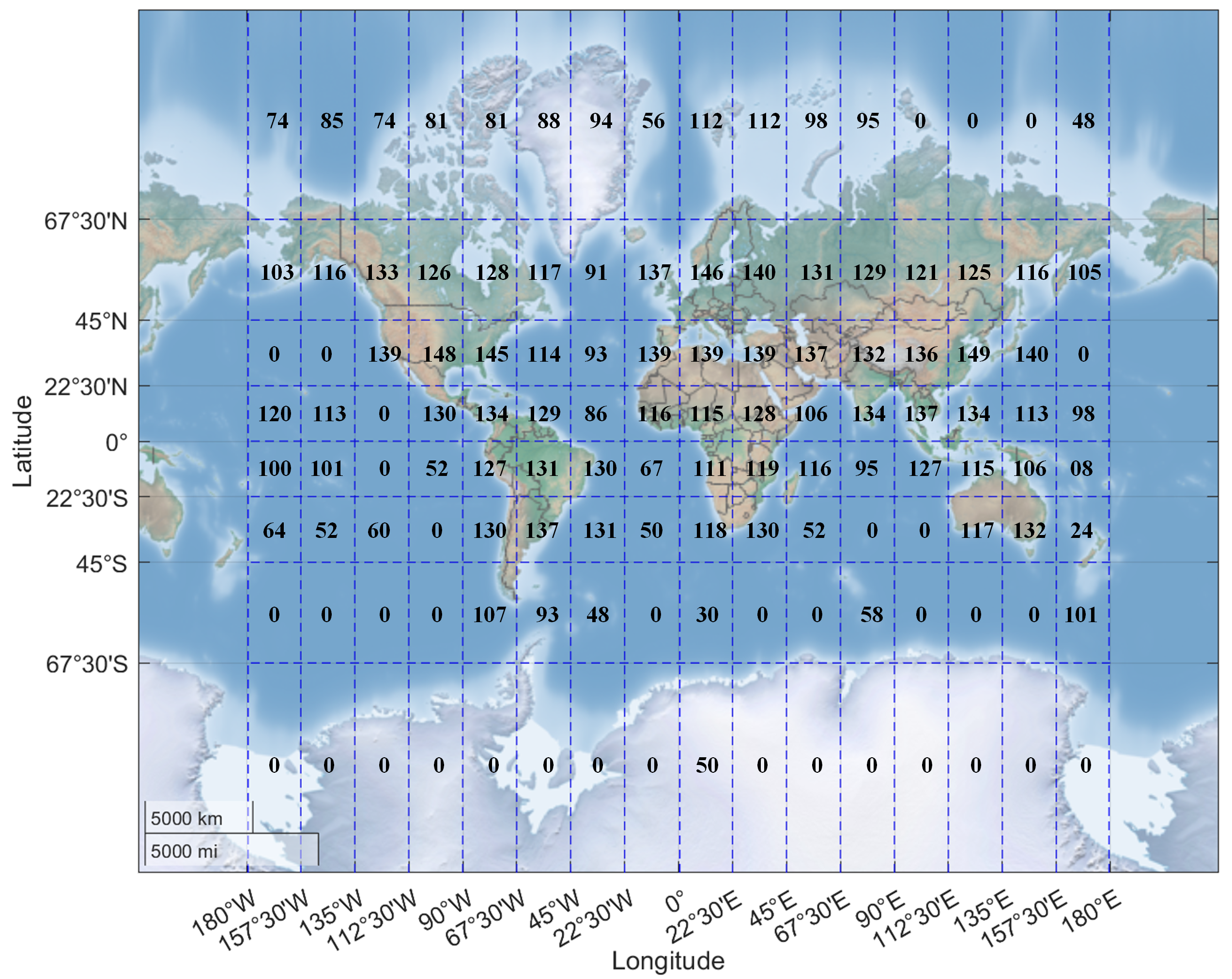}\\
	\begin{center}
		\caption{The connection index $ \eta $ of each zone.}\label{ConnectionIndex}
	\end{center}
	\vspace{-1.0em}
\end{figure}

\begin{definition}[Load]
	The load of a zone is denoted as the task arriving rate (i.e., the number of tasks arriving every second) of that zone.  
\end{definition}

The load of the whole network $ \mathcal{L} $ is the sum of the load of each zone (i.e., $\mathcal{L}=\sum\lambda_i$). When the network load $ \mathcal{L} $ is given, the load of a zone is proportion to the connection index of that zone (i.e., $ \lambda_i = \mathcal{L} \times \eta_i/\sum_{z\in{\tiny } \mathbf{Z}}{\eta_z} $).

\begin{figure*}[t]
	\centering
	\setlength{\abovecaptionskip}{-0.cm}
	\setlength{\belowcaptionskip}{-0.cm}
	\includegraphics[width=0.95\textwidth]{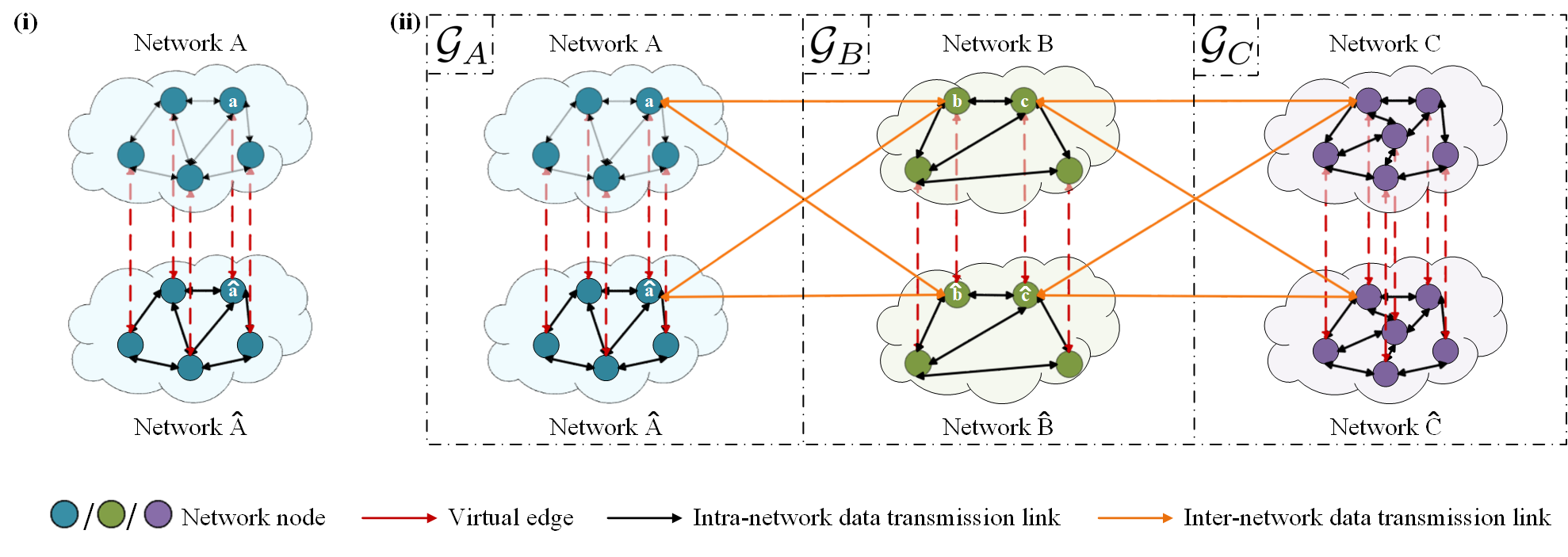}\\
	\begin{center}
		\caption{(i) Virtual-edge-based unify representation of transmission process and computation process. (ii) Metagraph-based computation and transmission fusion for multi-tier networks.  Network $ A $, $ B $ and $ C $ are three single-tier networks which can be most networks, such as satellite, UAV, IoT, and cellular networks. $ \mathcal{G}_A $, $ \mathcal{G}_B $ and $ \mathcal{G}_C $ are computation and transmission fused single-tier networks corresponding to network $ A $, $ B $ and $ C $, respectively.}\label{Multi-Tier_Multi-Process}
	\end{center}
	\vspace{-1.0em}	 
\end{figure*} 

\subsection{Delay Model}
In this paper, the \textit{overall} delay consists of the \textit{transmission} delay, the \textit{propagation} delay, the \textit{computation} delay and the \textit{waiting} delay.
  
In dynamic networks, the \textit{transmission} delay $T_{trans}$ of subtask $ \tau_l $ on edge $e(i,j)$ starting from instant $ t $ satisfies the following equation: $\widetilde{N_l}=\int_{t}^{t+T_{trans}} R_{i,j}^r(t) dt$, where $\widetilde{N_l} $ is the data volume of subtask $\tau_l$ and $R_{i,j}^r(t)$ is the available transmission rate of edge $ e=(i,j) $ at instant $ t $.

Similarly, the \textit{computation} delay $T_{comp}$ of $ \tau_l $ at node $i$ starting from instant $t$ satisfies the following equation: $\widetilde{C_l}=\int_{t}^{t+T_{comp}}C_i^r(t) dt$, where $\widetilde{C_l}$ is the computational requirement of subtask $\tau_l$. $C_i^r(t)$ is the amount of available computing capability that node $ i $ can provide at instant $ t $.

Ignoring the minor distance changes during transmissions, the \textit{propagation} delay $T_{prop}$ on edge $e(i,j)$ starting from instant $ t $ is mathematically defined as $T_{prop}(t)=D_{i,j}(t)/c$, where $ D_{i,j}(t) $ represents the distance between node $ i $ and node $ j $ at instant $ t $. $ c $ is the speed of light. 

The \textit{waiting} delay is the duration from the arrival of a subtask to the moment when the subtask starts being transmitted or processed. In the following sections, the waiting delay before transmission and computation are included in the corresponding transmission delay and computation delay, respectively.

\section{Multi-Tier Computation and Transmission Fusion Network Construction}\label{Metagraph-BasedFusionScheme}
In this section, a virtual-edge-based joint representation method of transmission and the computation processes is first proposed. Then a metagraph-based fusion network integration method for multi-tier networks is presented. 

\begin{definition}[Metagraph]
	A graph is denoted as $ G=\{\mathbf{V},\mathbf{E}\} $, where $ \mathbf{V} $ is the node set and $ \mathbf{E} $ is the edge set. Assuming that $ \{G_1,G_2,...,G_m\} (m\in Z^+) $ is a set of graphs. 
	A metagraph $ \mathbb{G} $ is a graphical representation of a set of graphs and the connections between them. It can be mathematically defined as $ \mathbb{G} = \{\mathbb{V},\mathbb{E}\} $, where $ \mathbb{V}=\{\mathbf{V}_1\cup\mathbf{V}_2\cup\dots\cup\mathbf{V}_m\} $, $ \mathbb{E}=\{\mathbf{E}_1\cup\mathbf{E}_2\cup\dots\cup\mathbf{E}_m\cup\mathbf{E}_{1,2}\cup\mathbf{E}_{2,3},\cup\dots\cup\mathbf{E}_{m-1,m}\} $. $ \mathbf{E}_{m-1,m} $ presents the set of edges connecting $ G_{m-1} $ and $ G_m $. For any $ e(i,j)\in \mathbf{E}_{m-1,m} $, $ i\in \mathbf{V}_m, j\in \mathbf{V}_{m-1}$ or $ i\in \mathbf{V}_{m-1}, j\in \mathbf{V}_m$.
\end{definition}

A metagraph can represent not only the internal connectivity of nodes within a graph, but also the connectivity between graphs. Based on this property, metagraphs are adopted to integrate multiple tiers of networks and fuse transmission and computation.

\subsection{Virtual-Edge-Based Multi-Process Representation}\label{Virtual-Edge-Constrution}
\subsubsection{\textbf{Motivation}}
Networks are often modeled as graphs, where network devices are modeled as nodes and transmission links are modeled as edges. In this way, problems in networks, such as routing and resource allocation, can be solved with graph-based algorithms.
To make these algorithms still applicable after fusing transmission and computation, the computation process is also expected to be represented by edges.
However, the computation process cannot be directly expressed as edges. It is because edges represent the connectivity between nodes, yet the computation process takes place on nodes rather than inter-node links.
Therefore, a new approach to abstract the computation process is needed.

\subsubsection{\textbf{Key Idea}}
This paper introduces \emph{virtual edges} to represent the computation process.
Virtual edges are built by generating a duplicate of the actual network and connecting each node in the actual network with the corresponding node in the duplicate.
To represent computational state migration with virtual edges, we assign the \textit{uncomputed state} and the \textit{computed state} to the two vertices of each virtual edge, respectively.
In this way, the transmission process (actual edges) and computation process (virtual edges) are uniformly represented, which enables the multi-process fusion in single-tier networks.

Despite the edge-based computation process representation, virtual edges will not cause conflicts in the physical meaning of the actual computation process.
Specifically, because both vertices of a virtual edge represent the same computing node, the computation process is still related only to the node performing that computation.


\subsubsection{\textbf{Implementation}}
Fig.~\ref{Multi-Tier_Multi-Process} (i) shows the virtual-edge-based joint representation of transmission process and computation process.
First, for any single-tier network $ A $, model this network as a graph which can be denoted as $ G_A=\big(\mathbf{V}_A,\mathbf{E}_A\big) $, where $ \mathbf{V}_A $ is the set of nodes and $ \mathbf{E}_A $ is the set of edges. Each edge $ e(i,j)\in \mathbf{E}_A$ represents the transmission link between node $ i $ and node $ j $.
Next, generate a duplicate of the original network $ G_A $. The network duplicate is denoted as $ G_{\widehat{A}}=\big(\mathbf{V}_{\widehat{A}},\mathbf{E}_{\widehat{A}}\big) $, where $ \mathbf{V}_{\widehat{A}}=\mathbf{V}_A $ and $ \mathbf{E}_{\widehat{A}}=\mathbf{E}_A $. \comment{I think the correct symbol of copied graph is \(G_{\widehat{A}}\), instead of \(G_{\widehat{A}}\). This also applies to E and V.} 
Then, connect each node in the original network $ G_A $ with the corresponding node in the copied network $ G_{\widehat{A}} $ to generate virtual edges whose set is $ \mathbf{E}_{virtual} $.
For any node $ i \in \mathbf{V}_A $,  
$ e(i,\widehat{i}) \in \mathbf{E}_{virtual} $ represents the virtual edge between node $ i \in \mathbf{V}_A $ and its duplicate $ \widehat{i} \in \mathbf{V}_{\widehat{A}} $. Each virtual edge $ e(i,\widehat{i}) \in \mathbf{E}_{virtual} $ is utilized to represent the computation process associated with node $ i $. Finally, a novel graph $ \mathcal{G}_A=\big(\mathcal{V}_A,\mathcal{E}_A\big)=\big(\mathbf{V}_A \cup \mathbf{V}_{\widehat{A}}, \mathbf{E}_A \cup \mathbf{E}_{\widehat{A}} \cup \mathbf{E}_{virtual}\big) $ is generated. In graph $ \mathcal{G}_A $, edges in $ G_A $ and $ G_{\widehat{A}} $ (i.e., $ \mathbf{E}_A \cup \mathbf{E}_{\widehat{A}} $) are utilized for representing the transmission process; whereas, virtual edges between $ G_A $ and $ G_{\widehat{A}} $ (i.e., $ \mathbf{E}_{virtual} $) are utilized for representing the computation process. 

\subsection{Metagraph-Based Multi-Tier Integration}\label{Metagraph-Based-NetworkIntergration}
\subsubsection{\textbf{Motivation}}
When integrating multiple single-tier networks into a multi-tier network, the inter-tier visibility should be considered.
It is because wireless network devices usually have limited communication distances. If two tiers of networks are distant from each other, such as satellite networks and terrestrial networks, the nodes of one network tier may not be fully visible (i.e., communicable) to the nodes of another network tier.

\subsubsection{\textbf{Key Idea}}
This paper introduces \emph{metagraphs} to integrate multiple single-tier fusion networks constructed in Section~\ref{Virtual-Edge-Constrution}.
Naturally, each single-tier network can be modeled as a graph; therefore, a fusion network can be modeled as \emph{graph of graphs}, where edges between graphs are used to represent inter-tier visibility.

The mathematical tool to describe graph of graphs is metagraph.
When appropriate physical constraints are assigned to the metagraph, mathematical characteristics of metagraph can be used to optimize computation offloading.
For example, when minimizing the delay, we can assign the corresponding sum of transmission delay and propagation delay to each actual edge, and then assign the corresponding computation delay to each virtual edge.
In this way, a path inside the metagraph can be projected to a concrete offloading decision.
The length of the path corresponds to the total delay of the decision.
Consequently, the joint optimization of transmission and computation for multi-tier network can be solved by searching the shortest path in the metagraph.



\subsubsection{\textbf{Implementation}}
Fig.~\ref{Multi-Tier_Multi-Process} (ii) presents a metagraph-based multi-tier computation and transmission fusion network. 
Inside Fig.~\ref{Multi-Tier_Multi-Process} (ii), Network $ A $, $ B $ and $ C $ are three single-tier networks which can be any \comment{Extreme words such as ``any'' ``all'' should be avoided.} network, such as satellite, UAV, IoT, and cellular networks. $ \mathcal{G}_A $, $ \mathcal{G}_B $ and $ \mathcal{G}_C $ are computation and transmission fused single-tier networks introduced in Section~\ref{Virtual-Edge-Constrution} corresponding to network $ A $, $ B $ and $ C $, respectively. 
Take the connections between $ \mathcal{G}_A $ and $ \mathcal{G}_B $ as an example. Assuming that node $ a (a\in \mathbf{V}_A)$ is visible to node $ b (b\in \mathbf{V}_B)$. That is, $ a $ and $ b  $ form a visible node pair.
There are 4 edges that can be added, i.e., $ e(a,b) $, $ e(a,\widehat{b}) $, $ e(\widehat{a},b) $ and $ e(\widehat{a},\widehat{b}) $.
Each visible node pair between $ \mathcal{G}_A $ and $ \mathcal{G}_B $ needs to conduct the above connections. In this way, a metagraph, which contains multiple single-tier networks and fuses computation and transmission, is generated.

Specific offloading strategies can be decoded from the path-finding results. For example, assuming that the shortest path from node $ a $ to $ c $ is $ a \rightarrow \widehat{a} \rightarrow \widehat{b} \rightarrow b \rightarrow c $. $ a \rightarrow \widehat{a} $ and $ \widehat{b} \rightarrow b $ indicates that the task should be processed on both node $ a $ and $ b $. $ \widehat{a} \rightarrow \widehat{b} $ indicates inter-tier data transmission; whereas $ b \rightarrow c $ indicates intra-tier data transmission.

\section{Network-Wide Task Offloading With Multi-Tier Fusion Network}\label{AlgorithmSteps}
In this section, we take the network-wide offloading problem for space missions presented in Section~\ref{Introduction} as an example to illustrate how to construct a multi-tier fusion network based on real-world scenarios and achieve the network-wide optimal computation offloading. Then, we prove the superior performance of the proposed computation offloading scheme over the conventional schemes. The proof is based on a graph theory performed on metagraphs, where the conventional task offloading schemes are also converted to metagraphs.

\subsection{Network-Wide Computation Offloading for Space Missions}\label{FusionScheme}
The network-wide computation offloading problem for space missions proposed in Section~\ref{Introduction} involves three single-tier networks, i.e., the RS satellite network, the LEO satellite network and the terrestrial network. To achieve network-wide computation offloading, first, the above multi-tier network is constructed as a metagraph based on the multi-tier fusion network construction method proposed in Section~\ref{Metagraph-BasedFusionScheme}. Then, edges of the metagraph are assigned with appropriate weights. At last, the optimal offloading path in the metagraph is searched using the shortest path algorithm. Details are described below.

\begin{figure}[b]
	\centering
	\setlength{\abovecaptionskip}{-0.cm}
	\setlength{\belowcaptionskip}{-0.cm}
	\includegraphics[width=0.45\textwidth]{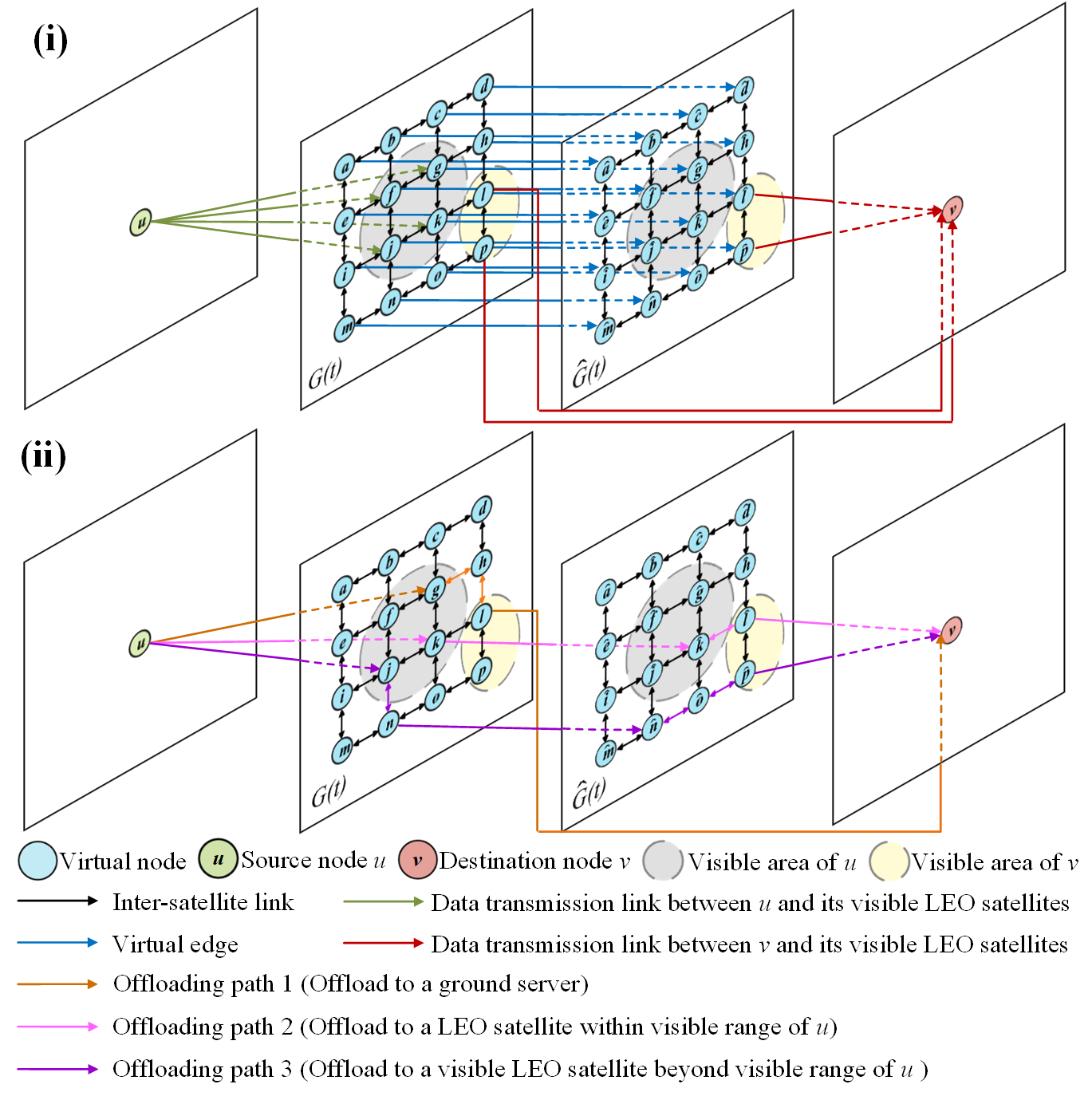}\\
	\begin{center}
		\caption{(i) The metagraph of fusion offloading scheme for space missions. (ii) Typical offloading paths.}\label{FusionNetwork}
	\end{center}	 
\end{figure}

Fig.~\ref{FusionNetwork} shows the metagraph-based multi-tier fusion network and typical offloading paths of the network-wide offloading scheme for space missions. For clarity, only one source node and one destination node are drawn. The properties of this three-tier network are summarized as follows. 
\begin{itemize}
	\item \textbf{Network tiers:} the RS satellite network, the LEO satellite network and the terrestrial network.
	\item \textbf{Source of tasks:} RS satellites.
	\item \textbf{Destination of tasks:} ground servers.
	\item \textbf{Offloading targets:} LEO satellites or Ground servers.
	\item \textbf{Assumptions:} the computing capabilities of ground servers are strong enough that the computation delay of tasks can be ignored when they are offloaded to ground servers for computation.
\end{itemize}

As shown in Fig.~\ref{FusionNetwork} (i), tasks are offloaded from the source node $ u $ to the destination node $ v $ via the LEO satellite network.
To enable network-wide offloading, the computation and transmission fusion network according to this scenario is constructed. The construction steps are stated as follows. 

\textbf{STEP 1.1:}  Model the dynamic LEO satellite network as a virtual-node-based network (introduced in Section~\ref{Network Model}) $ G(t)=\big(\mathbf{V},\mathbf{E}\big) $\footnote{The LEO satellite network is defined as $ G(t) $ because it is a dynamic network. After being modeled as a virtual-node-based network, the dynamics is mainly reflected in the change of the association between LEO satellites and virtual nodes, so the node set and the edge set are denoted as time-independent variables.}, where $ \mathbf{V} $ is the set of virtual nodes (hereinafter referred to as nodes) and $ \mathbf{E} $ is the set of edges. Each edge $ e(i,j)\in \mathbf{E} $ represents the transmission process on the ISL between the LEO satellites associated with node $ i $ and node $ j $. 

\textbf{STEP 1.2:} Generate a network duplicate $ \widehat{G}(t)=\{\widehat{\mathbf{V}},\widehat{\mathbf{E}}\} $ of the original virtual-node-based network $ G(t) $.

\textbf{STEP 1.3:} Generate virtual edges as Section~\ref{Virtual-Edge-Constrution} stated. For each $ i\in \textbf{V} $, connect node $ i $ and the corresponding node $ \widehat{i} $ in $ \widehat{G}(t) $ with a virtual edge. The set of virtual edges is $ \textbf{E}_{virtual} $.
 
\textbf{STEP 1.4:} Connect $ u $ and nodes whose associated satellites are visible to $ u $ in $ G(t) $ such that raw data of tasks can be transmitted from the source node $ u $ to the LEO satellite network. The set of edges connecting $ u $ and its visible nodes is defined as $ \textbf{E}_{u,G(t)} $.

\textbf{STEP 1.5:} Connect $ v $ and nodes whose associated satellites are visible to $ v $ in both $ G(t) $ and $ \widehat{G}(t) $ such that raw data or computing results of tasks can be offloaded from the LEO satellite network to the destination node $ v $. The set of edges connecting $ v $ and its visible nodes is defined as $ \textbf{E}_{v,G(t),\widehat{G}(t)} $. 

\textbf{STEP 1.6:} For every subtask $ \tau_l $, assign each $ e\in \mathbf{E}\cup\widehat{\mathbf{E}}$ with the sum of its transmission delay $ T_{trans} $ and propagation delay $ T_{prop} $.

\textbf{STEP 1.7:} For every subtask $ \tau_l $, assign each $ e\in \textbf{E}_{virtual}$ with its computation delay $ T_{comp} $.

After performing the steps above, the metagraph $ \mathbb{G}_F=\big(\mathbb{V}_F,\mathbb{E}_F\big)=\big(\mathbf{V}\cup\widehat{\mathbf{V}}\cup u \cup v,\textbf{E}\cup\widehat{\textbf{E}}\cup \textbf{E}_{u,G(t)}\cup\textbf{E}_{v,G(t),\widehat{G}(t)}\big)$\footnote{The edges in LEO satellite network $ G(t) $ and its duplicate $ \widehat{G}(t) $ are bidirectional, which enables tasks to be offloaded to any LEO satellite for computation; whereas the inter-tier edges are unidirectional, which leads to each path containing at most one virtual edge (i.e., each subtask being processed by at most one LEO satellite).} under the constraints of the fusion offloading scheme is obtained. The graph is used for finding the offloading path from $ u $ to $ v $.

Depending on the edges contained in the offloading path, tasks can be offloaded to ground servers, visible LEO satellites or invisible LEO satellites for computation. Fig.~\ref{FusionNetwork} (ii) illustrates three typical scenarios.

\begin{itemize}
	\item \textit{Path 1} ($ u \rightarrow g \rightarrow h \rightarrow l \rightarrow v $) offloads tasks to ground servers.
	\ding{172} Raw data of tasks are transmitted from source node $ u $ to node $ g $ which locates within the visible range of $ u $.
	\ding{173} Raw data of tasks are routed from $ g $ to $ l $ in $ G(t) $ through ISLs.
	\ding{174} Raw data of tasks are transmitted from $ l $ in $ G(t) $ to the destination node $ v $ on the ground for computation.
	In general, tasks are offloaded to the ground when the offloading path does not contain any virtual edges, since virtual edges represent the onboard computing process.
	
	\item \textit{Path 2} ($ u \rightarrow k \rightarrow \widehat{k} \rightarrow \widehat{l} \rightarrow v $) offloads tasks to visible LEO satellites.
	\ding{172} Raw data of tasks are transmitted from source node $ u $ to node $ k $ which locates within the visible range of $ u $.
	\ding{173} Edge $ k \rightarrow \widehat{k} $ indicates that tasks are computed by the LEO satellite associated with node $ k $.
	\ding{174} The computing results are routed from $ \widehat{k} $ to $ \widehat{l} $ in $ \widehat{G}(t) $.
	\ding{175} The computing results are further transmitted from $ \widehat{l} $ to the destination node $ v $ on the ground.
	In general, tasks are offloaded to the corresponding visible LEO satellite, when the offloading path contains a virtual edge which connects one of the visible nodes of $ u $.
	
	\item \textit{Path 3} ($ u \rightarrow j \rightarrow n \rightarrow \widehat{n} \rightarrow \widehat{o} \rightarrow \widehat{p} \rightarrow v $) offloads tasks to invisible LEO satellites.
	\ding{172} Raw data of tasks are transmitted from $ u $ to its visible VN $ j $.
	\ding{173} Raw data of tasks are routed from $ j $ to $ n $ in $ G(t) $ through ISLs.
	\ding{174} Edge $ n \rightarrow \widehat{n} $ indicates that tasks are computed on satellites associated with $ n $.
	\ding{175} The computing results are routed from $ \widehat{n} $ to $ \widehat{p} $ in $ \widehat{G}(t) $.
	\ding{176} The computing results are further transmitted from $ \widehat{p} $ to the destination $ v $ on the ground. Because nodes in $ G(t) $ are connected with ISLs, tasks can be offloaded to satellites associated with any VN for computing.
	In general, tasks are offloaded to an invisible LEO satellite, when the offloading path contains a virtual edge but the virtual edge does not connect to one of the visible nodes of $ u $.
\end{itemize}



Among many viable offloading decisions of the three scenarios, the optimal one can be obtained by finding the shortest path in the weighted metagraph. Algorithm~\ref{algo:meta} presents the detailed steps for instantiating the graph, assigning computation and transmission delays, and obtaining the optimal decision with Dijkstra's single-source shortest path algorithm.

\begin{algorithm}[h!]
	\caption{Metagraph-based computation and transmission fusion task offloading scheme\label{algo:meta}}
	\KwData{Network resourece set ($ \mathbf{C}, \mathbf{R}_{ISL}, \mathbf{R}_{SGL} $)}
	\KwData{Task $ \Lambda_ {\mathcal{T}} = (u,v,n,t,\vartheta) $}
	\KwData{Subask $ \Lambda_ {\tau_l} = (\widetilde{C_l},\widetilde{N_l}) $}
	\KwResult{The optimal offloading path $ \pi $ and its length $ L $}
	\KwResult{Updated network resourece set ($ \mathbf{C}, \mathbf{R}_{ISL}, \mathbf{R}_{SGL} $)}

	Generate the VN-based network model $ G(t)=(\mathbf{V},\mathbf{E}) $ for the LEO satellite network;
	
	Initialize the virtual edge set $ \mathbf{E}_{virtual} \leftarrow \emptyset $ and the weighted metagraph $ \mathbb{G}_F= (\mathbb{V}_F,\mathbb{E}_F,\mathbb{W}_F) \leftarrow (\emptyset,\emptyset,\emptyset)$;
	
	Update $ \mathbb{V}_F \leftarrow \mathbf{V},\ \mathbb{E}_F \leftarrow \mathbf{E}$;
	
	Generate the network duplicate $\widehat{G}(t) \leftarrow (\widehat{\mathbf{V}},\widehat{\mathbf{E}})\leftarrow G(t) $;
	
	Update $ \mathbb{V}_F\leftarrow \mathbb{V}_F \cup \widehat{\mathbf{V}} ,\  \mathbb{E}_F\leftarrow \mathbb{E}_F \cup \widehat{\mathbf{E}}$;
	
	Update $ \mathbb{V}_F \leftarrow \mathbb{V}_F\cup u\cup v $;
	
	\ForEach{$ i\in \mathbf{V} $}{
	   $ \mathbb{E}_F\leftarrow \mathbb{E}_F \cup e(i,\widehat{i}),\ \mathbf{E}_{virtual} \leftarrow \mathbf{E}_{virtual}\cup e(i,\widehat{i})$;}
    
    \ForEach{$ i\in \mathbf{V} $}{
   		\If{$ i \Leftrightarrow u $}{$  \mathbb{E}_F\leftarrow \mathbb{E}_F \cup e(u,i)  $;}
   		\If{$ i \Leftrightarrow v $}{$  \mathbb{E}_F\leftarrow \mathbb{E}_F \cup e(i,v) \cup e(\widehat{i},v)  $;}
   	}    

    \ForEach{$ \tau_l \in \mathcal{T}  $}{
    	\ForEach{$ e\in \mathbf{E}_{virtual} $}{
    		Calculate the computation delay $ T_{comp} $ of processing $  \tau_l $ on $ e $ with $\mathbf{C}$;
    		
	        Set $ \omega_{e} \leftarrow T_{comp},\ \mathbb{W}_F \leftarrow\mathbb{W}_F\cup \omega_{e} $;
    }
		\ForEach{$ e\in \mathbb{E}_F\backslash \mathbf{E}_{virtual} $}{
	        Calculate the transmission delay $ T_{trans} $ of transmitting $  \tau_l $ on $ e$ with $\mathbf{R}_{ISL}$ or $\mathbf{R}_{SGL}$ ;
	        
	        Set $ \omega_{e} \leftarrow T_{trans}+T_{prop} ,\ \mathbb{W}_F \leftarrow\mathbb{W}_F\cup \omega_{e} $;
    }
        Find the shorest path $ \pi$ and its length $ L $ in $ \mathbb{G}_F=\{\mathbb{V}_F,\mathbb{E}_F,\mathbb{W}_F\} $ with shortest path algorithms;
        
        \If{$ l \leq \vartheta $}{
           Output the optimal offloading path $ \pi $ and the path length $ L $ for subtask $ \tau_l $;
           
           Update the resource matrices $ \mathbf{C} $, $ \mathbf{R}_{ISL} $ and  $ \mathbf{R}_{SGL} $;}
    }
\end{algorithm}

Line 1--17 of Algorithm~\ref{algo:meta} presents detailed steps for constructing a metagraph.
The algorithm first initializes necessary variables (line 1--3) and copies the original network into it.
Next, the algorithm generates the duplicate to encode the computation state change (line 4 and 5).
With multiple base tiers available, the algorithm generates edges between the tiers to form a complete metagraph: the source and destination nodes (line 6), the virtual edges (line 7--9), and the visible nodes of the source/destination nodes (line 11--13 and line 14--16, respectively).

Line 18--32 of Algorithm~\ref{algo:meta} presents the offloading strategy generation using path-finding with the generated metagraph.
The main loop iterates each task through all pending tasks.
For each task, the algorithm first computes the metagraph's edge weights (line 19--26) according to task properties and available resources: for virtual edges (used to represent computation), the weights are dependent on the computation delay (line 19--22); similarly, for concrete edges (used to represent transmission), the weights are dependent on the transmission delay (line 23--26).
After the weights are assigned for all edges, the existing shortest path algorithm is invoked to find the optimal path and delay (line 27). 
Finally, the algorithm checks whether the task can be completed within the given delay threshold (line 28).
In this case, the task can be executed, and the optimal offloading strategy is output according to the shortest path (line 29). Since executing the task requires resources, line 30 of the algorithm updates the matrices to reflect the change.

\subsection{Metagraph-Based Conventional Task Offloading}\label{MetaOfConventionalScheme}
Metagraphs possess favorable abstraction and expressiveness capability. It is also capable of representing other task offloading schemes. 
To compare the performance of the proposed offloading scheme with the conventional offloading schemes introduced in Section~\ref{ConventionalScheme} theoretically, we construct the metagraphs of these conventional task offloading schemes and present them in Fig.~\ref{ConventionalScheme}.
\begin{figure}[h]
	\centering
	\setlength{\abovecaptionskip}{-0.cm}
	\setlength{\belowcaptionskip}{-0.cm}
	\includegraphics[width=0.45\textwidth]{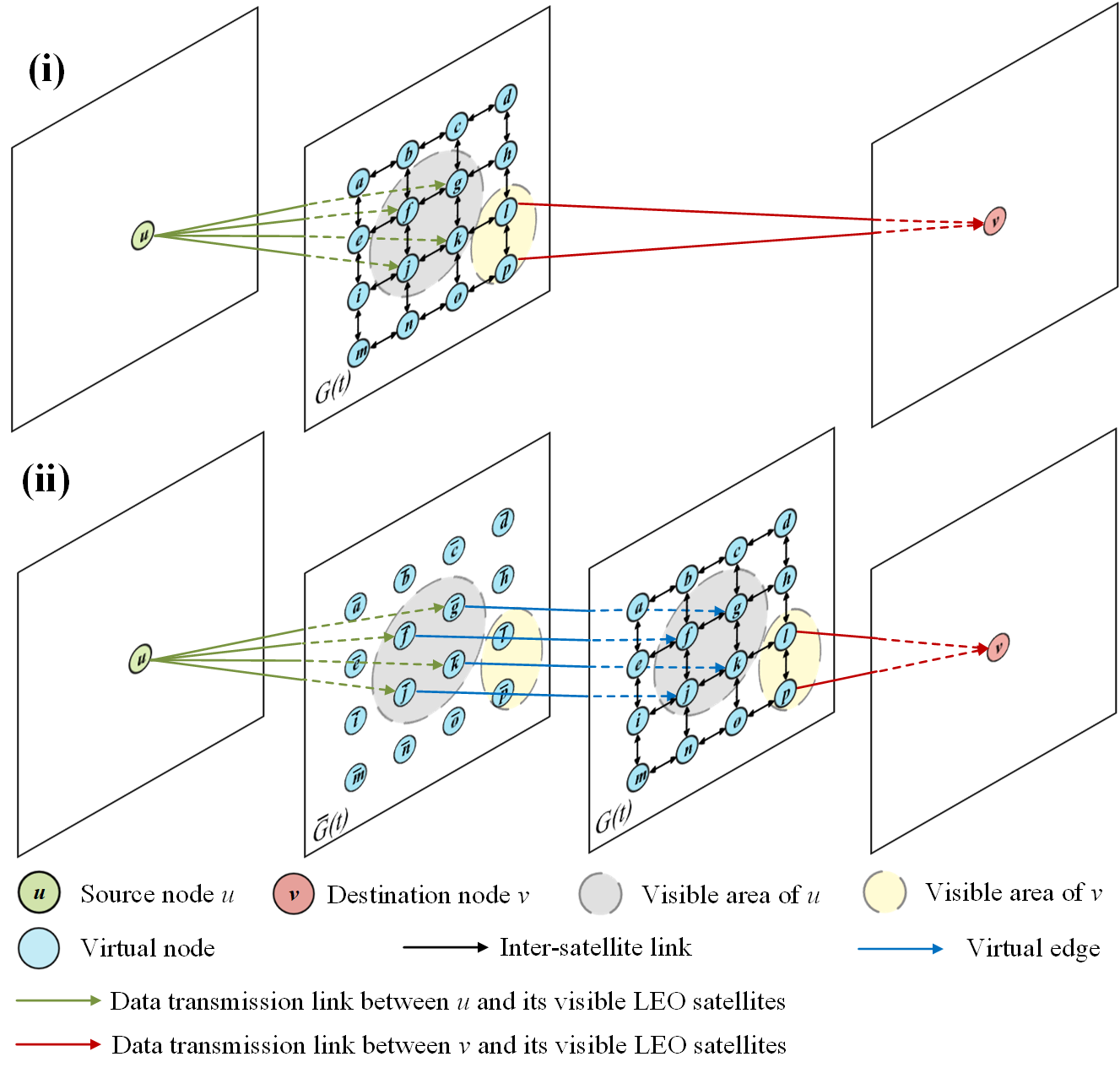}\\
	\begin{center}
		\caption{Metagraph representation for conventional task offloading schemes. (i): ground-based offloading; (ii): visible-LEO-satellite-based offloading.}\label{ConventionalScheme}
	\end{center}	 
\vspace{-0.5em}
\end{figure}
Figure~\ref{ConventionalScheme} (i) illustrates the ground offloading scheme, where raw data of tasks are transmitted via the LEO satellite network to ground servers for computation. Formally, the metagraph for this scheme is constructed in the following three steps:

\textbf{STEP 2.1:} The same as STEP 1.1 and STEP 1.4.

\textbf{STEP 2.2:} Connect $ v $ and nodes whose associated satellites are visible to $ v $ in $ G(t) $. The set of edges connecting $ v $ and its visible nodes is defined as $ \textbf{E}_{v,G(t)} $. 

\textbf{STEP 2.3:} The same as STEP 1.6.

After the construction, the metagraph under the constraints of the ground offloading scheme is $ \mathbb{G}_{G}=\big(\mathbb{V}_{G},\mathbb{E}_{G}\big)=\big(\mathbf{V}\cup u \cup v,\textbf{E}\cup\textbf{E}_{u,G(t)}\cup\textbf{E}_{v,G(t)}\big)$.


Figure~\ref{ConventionalScheme} (ii) illustrates the visible offloading scheme. In this scheme, tasks are offloaded to neighboring visible LEO satellites for computation; next, the computing results are routed to the ground as the destination. Formally, the metagraph for this scheme is constructed in the following five steps:

\textbf{STEP 3.1:} The same as STEP 1.1.

\textbf{STEP 3.2:} Generate a duplicate of $ G(t) $ and remove all edges (i.e., ISLs) in this duplicate\footnote{The reason for removing the edges in the duplicate is that tasks can only be transmitted to the visible satellites and are not allowed to be transmitted in the LEO satellite network until these tasks have been processed.}. In this way, a novel graph $ \bar{G}(t) = (\bar{\textbf{V}},\bar{\textbf{E}}) $ is formed, where $ \bar{\textbf{V}}=\textbf{V} $ and $ \bar{\textbf{E}} =\emptyset $.

\textbf{STEP 3.3:} Connect each visible node of $ u $ in $ \bar{G}(t) $ and its corresponding node in $ G(t) $ with virtual edges whose set is $ \mathbf{E}_{virtual}^{'} $.

\textbf{STEP 3.4:} Connect $ u $ with its visible nodes in $ \bar{G}(t) $, which allows data transmission of tasks from $ u $ to its visible LEO satellites.

\textbf{STEP 3.5:} The same as STEP 2.2, 1.6, and 1.7.

After the construction, the metagraph under the constraints of the visible offloading scheme is $ \mathbb{G}_{V}=\big(\mathbb{V}_{V},\mathbb{E}_{V}\big)=\big(\mathbf{V}\cup \bar{\textbf{V}} \cup u \cup v,\textbf{E}\cup \mathbf{E}_{virtual}^{'} \cup \textbf{E}_{u,\bar{G(t)}}\cup\textbf{E}_{v,G(t)}\big)$. .

\subsection{Performance Comparison}
Compare Fig~\ref{ConventionalScheme} (i) and Fig~\ref{FusionNetwork} (i), and we can conclude that $ \mathbb{G}_G $ --- the metagraph of the \textit{ground offloading scheme} --- is a subgraph of $ \mathbb{G}_F $, the metagraph of the proposed \textit{fusion offloading scheme} scheme. 
Similarly, $ \mathbb{G}_V $ (the metagraph of the \textit{visible offloading scheme} shown in Fig~\ref{ConventionalScheme} (ii)) is also a subgraph of $ \mathbb{G}_F $. Therefore, Theorem~\ref{theorem1} can be obtained.

\begin{theorem}
	For any task, the minimum overall delay that can be obtained by the fusion offloading scheme always does not exceed the minimum delay that can be obtained by the visible offloading scheme and the ground offloading scheme.\label{theorem1}
\end{theorem}

\begin{proof}
	Assuming that $ \mathbb{G}_F=(\mathbb{V}_F,\mathbb{E}_F) $, $ \mathbb{G}_G=(\mathbb{V}_G,\mathbb{E}_G) $ and $ \mathbb{G}_V=(\mathbb{V}_V,\mathbb{E}_V) $ are the metagraphs of the proposed \textit{fusion offloading scheme}, the \textit{ground offloading scheme} and the \textit{visible offloading scheme}, respectively. 
	Based on the construction methods introduced in Section~\ref{FusionScheme} and Section~\ref{MetaOfConventionalScheme}, it can be concluded that $ \mathbb{V}_F=\mathbb{V}_G=\mathbb{V}_V $, $ \mathbb{E}_G \subset \mathbb{E}_F $ and $ \mathbb{E}_V \subset \mathbb{E}_F $. Therefore, $ \mathbb{G}_G $ and $ \mathbb{G}_V $ are subgraphs of $ \mathbb{G}_F $. 
	
	For any task whose source node is $ u (u\in\mathbb{V}_F) $ and destination node is $ v(v\in\mathbb{V}_F) $, assuming that $ \mathbb{P}_{F} $,  $ \mathbb{P}_{G} $ and $ \mathbb{P}_{V} $ are the set of paths from $ u $ to $ v $ in $ \mathbb{G}_F $, $ \mathbb{G}_G $ and $ \mathbb{G}_V $, respectively. 
	The optimal paths that can obtain the minimum delay in $ \mathbb{G}_F $, $ \mathbb{G}_G $ and $ \mathbb{G}_V $ are defined as $ \pi_{F} $, $ \pi_{G} $ and $ \pi_{V} $, and their path lengths are $ l_{F} $, $ l_{G} $ and $ l_{V} $, respectively. 
	
	Since $ \mathbb{E}_G \subset \mathbb{E}_F$ and $ \pi_{G} $ is a path in $ \mathbb{G}_G $, $ \pi_{G} $ is a path in $ \mathbb{G}_F $ (i.e., $ \pi_{G} \in \mathbb{P}_{F} $). 
	Since $ \pi_{F} $ is the optimal path of $ \mathbb{G}_F $, for any path $ p \in \mathbb{P}_{F} $ whose length is $ l_p $, $ l_{F} \leq l_p $. Because $ \pi_{G} \in \mathbb{P}_{F} $, we can get $ l_{F} \leq l_{G}$. The relation between $ l_{F}$ and $ l_{V}$ can be proofed in a similar way.
	$\hfill\blacksquare$
\end{proof}
\section{Simulation Results and Analyses}
In this section, we evaluate the metagraph-based fusion offloading scheme proposed in Section~\ref{AlgorithmSteps} whose performance is compared with two typical benchmark schemes introduced in Section~\ref{Introduction} (i.e., the ground offloading scheme and the visible offloading scheme). 
We would like to emphasize that the simulations in this paper are focused on scheme-level comparison, which does not involve detailed comparison on individual algorithms. 
Because individual algorithms' task division methods are crucial to the final results, we control this variable by assuming the smallest unit of transmission and computation to be a subtask.

We aim to answer the following research questions.

\begin{itemize}
	\item \textbf{RQ1}: Does the proposed scheme reduce the overall delay?
	\item \textbf{RQ2}: How applicable is the proposed scheme, especially for systems with low computation and transmission capabilities?
	\item \textbf{RQ3}: Is the performance of the proposed scheme consistent for tasks with different computational requirements or data volumes?
\end{itemize}

The parameter settings and simulation results are elaborated as follows.
A Walker constellation with 8 orbits is adopted in this section. There are 16 satellites evenly distributed on each orbit.
Network-related parameters are set based on the-state-of-the-art LEO satellite techniques introduced in Section~\ref{Background}. We simulate a 10-second network task generation. Other task-related parameters follow Section~\ref{TrafficModel}. Since images generated by satellite sensors can reach Gigabit at most~\cite{8896440}, the data volume of each subtask is set as 0.1 GB.\ The energy and storage resources of LEO satellites are assumed to be sufficient, because they are related to the specific hardware devices and algorithms. We leave them to future work.
The basic parameters are summarized in Table~\ref{T1}. If not stated additionally, the following experiments are conducted according to the parameters listed in Table~\ref{T1}.
The core metric used for evaluation is weighted average delay.

\begin{table}[h]
	\centering
	\begin{spacing}{1}
		\caption{Simulation Parameters}\label{T1}
	\end{spacing}
	\small
	\begin{tabular}{cc}
		\toprule
		Parameter & Value\\
		\midrule
		Radius of the earth $ R_e $                  & 6,371,393 m\\
		Mass of the earth $ M_e $                    & $ 5.965\times10^{24}$ kg\\
		Earth rotation angular velocity $ \omega_e $ & $ 7.29211510\times 10^{-5} $ \\
		Gravitational  $ G $                         & $ 6.67428\times 10^{-11} $\\
		Kepler constant $ K $                        & $ 3.9860\times 10^{14} $\\
		Velocity of light $ c $                      & $ 299,792,458$ m/s\\	
		Number of orbits $ N_o $                     & 8\\
		Satellites per orbit $ N_s $                 & 16\\        
		Orbit altitude $ h $                         & 500 km\\
		Orbit inclination $ i_0 $                    & $ 90^{\circ} $\\
		Max computing capability of satellites $ C $ & 100 GFLOPS\\    
		Max data rate of SGL                        & 0.2 Gbps\\
		Max data rate of ISL                     & 5 Gbps\\
		Channel number per ISL                       & 1 \\
		Network load $ \mathcal{L} $                 & 200 \\
		Task generation duration                     & 10s\\
		Subtask number per task $ \widetilde{Q} $    & 2\\
		Computational requirement per subtask $ \widetilde{C} $ & 100 GFLO\ \\
		Data volume per subtask $ \widetilde{N} $    & 0.1 Gbps\\
		Delay threshold $ \theta_{th} $              & 300 s \\	
		\bottomrule
	\end{tabular}
\end{table}


\begin{definition}[Weighted average delay]
	The weighted average delay is the ratio of the sum overall delay of all tasks generated in all zones during the simulation time to the total number of tasks, where the overall delay of any task that fail to offload is set to the delay threshold $ \theta_{th} $.
\end{definition}
%
\subsection{Reduced Overall Delay (RQ1)}
To demonstrate how the proposed scheme reduces the overall delay, we evaluate all schemes under different system loads ($ \mathcal{L} $) and present the weighted average delay in Fig.~\ref{Change_network_load}.
Due to the large delay of the ground offloading scheme, the data on the curve corresponding to this scheme was set to 1/5 of the actual values for clear presentation.

\begin{figure}[htbp]
	\centering
	\vspace{-1.0em}
	\setlength{\abovecaptionskip}{-0.cm}
	\setlength{\belowcaptionskip}{-0.cm}
	\includegraphics[width=0.5\textwidth]{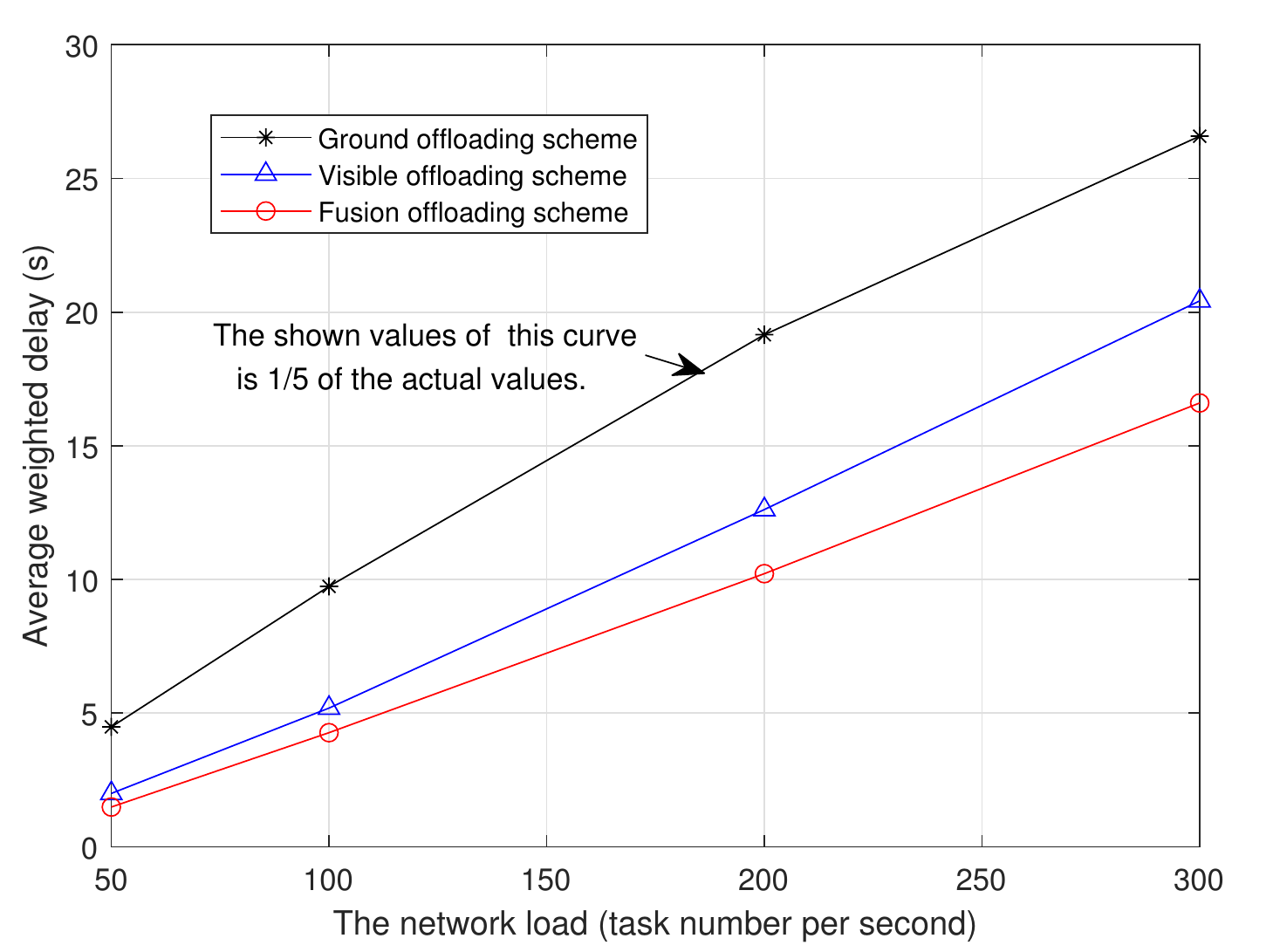}\\
	\begin{center}
		\caption{The weighted average delays versus the network load $ \mathcal{L} $.}\label{Change_network_load}
	\end{center}
\vspace{-1.0em}
\end{figure}

As shown in Fig.~\ref{Change_network_load}, the weighted average delays of all offloading schemes increase when the network loads become progressively heavier. As the number of tasks to be offloaded per second increases, the remaining transmission and computational resources in the network keep decreasing, so more time is needed to complete task offloading.
However, despite the increase in delays, the proposed fusion offloading scheme consistently achieves the lowest weighted average delay under all simulated loads. In particular, the proposed scheme decreases the weighted average delay up to 87.51\% and 18.70\% compared with the ground offloading scheme and the visible offloading scheme when $ \mathcal{L}=300 $. The reasons behind this can be explained by the flexible offloading target selection of the proposed scheme. More specifically, the proposed scheme is able to utilize the computing resources on invisible satellites to overcome the limitations of the benchmark schemes. Moreover, it allows the switch among different types of offloading targets based on the network load. That is, the proposed scheme will select ground servers or visible satellites to conduct computation between if they can achieve the minimum overall delay. Consequently, the proposed fusion offloading scheme can always achieve the minimum overall delay.

Inside Fig.~\ref{Change_network_load}, the growth rate of these offloading schemes are significantly different. Specifically, the curve of the proposed fusion offloading scheme has the lowest growth rate. It is because this scheme can offload tasks to invisible LEO satellites for computation. The fusion offloading scheme solves the problem of limited resources available on neighboring satellites in the visible offloading scheme and avoids the problem of excessive transmission delay caused by routing large amounts of raw data to ground servers in the ground offloading scheme.

To understand why the proposed scheme can lower the overall delay, Fig.~\ref{DelayComponent} investigates the detailed components of the overall delay of these offloading schemes when $ \mathcal{L}=200 $.
For fairness, the data is collected under the following principle: only when the fusion scheme selects to offload tasks to invisible satellites, the delay components of the fusion, ground and visible offloading scheme are counted. It is worth noting that to better distinguish the delay of the transmission process and the computation process, the transmission delay and the computation delay in Fig.~\ref{DelayComponent} include the waiting delay for the start of transmission or computation after reaching the corresponding node.
It is because when the fusion scheme offloads tasks to ground servers and visible satellites, the delay components of this scheme are exactly the same as two benchmark schemes.

\begin{figure}[h]
	\centering
	\vspace{-0.5em}
	\setlength{\abovecaptionskip}{-0.cm}
	\setlength{\belowcaptionskip}{-0.cm}
	\includegraphics[width=0.5\textwidth]{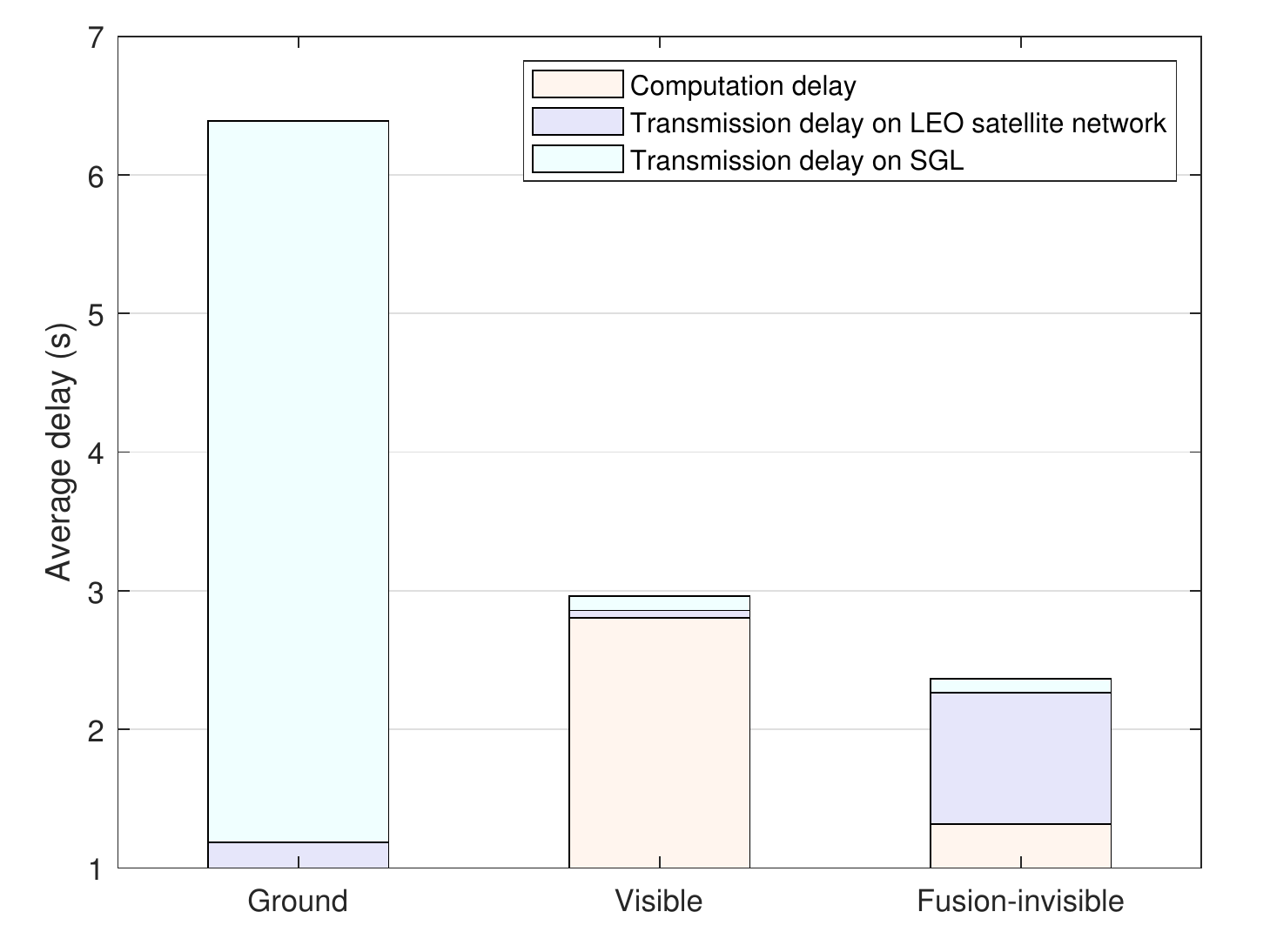}\\
	\begin{center}
		\caption{Breakdown of the overall delay when $ \mathcal{L}=200 $.}\label{DelayComponent}
	\end{center}
\vspace{-1.5em}
\end{figure}

According to Fig.~\ref{DelayComponent}, it can be concluded that both the ground offloading scheme and the visible offloading scheme have delay bottlenecks. For the ground offloading scheme, the transmission delay on SGLs is much higher than the other two schemes. It is because a large amount of raw data is transmitted on SGLs whose data transmission rate is very limited. The other two schemes only transmitted the computing results (usually a few bites) to the ground. For the visible offloading scheme, the computation delay (including the corresponding delay of waiting the computing resources available) is much higher than the other two schemes. It is because this scheme only offloaded tasks to visible LEO satellites; therefore, the available computing resources are quite limited. Tasks need to wait a long time before they can be processed. The proposed fusion scheme avoids the bottlenecks of the two benchmark schemes. The proposed scheme is able to balance the network load excellently, which is extremely important to hotspots with heavy loads.



To illustrate how the proposed scheme flexibly changes the offloading target according to the load, we present a concrete example in Fig.~\ref{Target_Distribution}. It shows the distribution of offloading targets for tasks generated in a selected hotspot. 

\begin{figure}[t]
	\centering
	\includegraphics[width=0.5\textwidth]{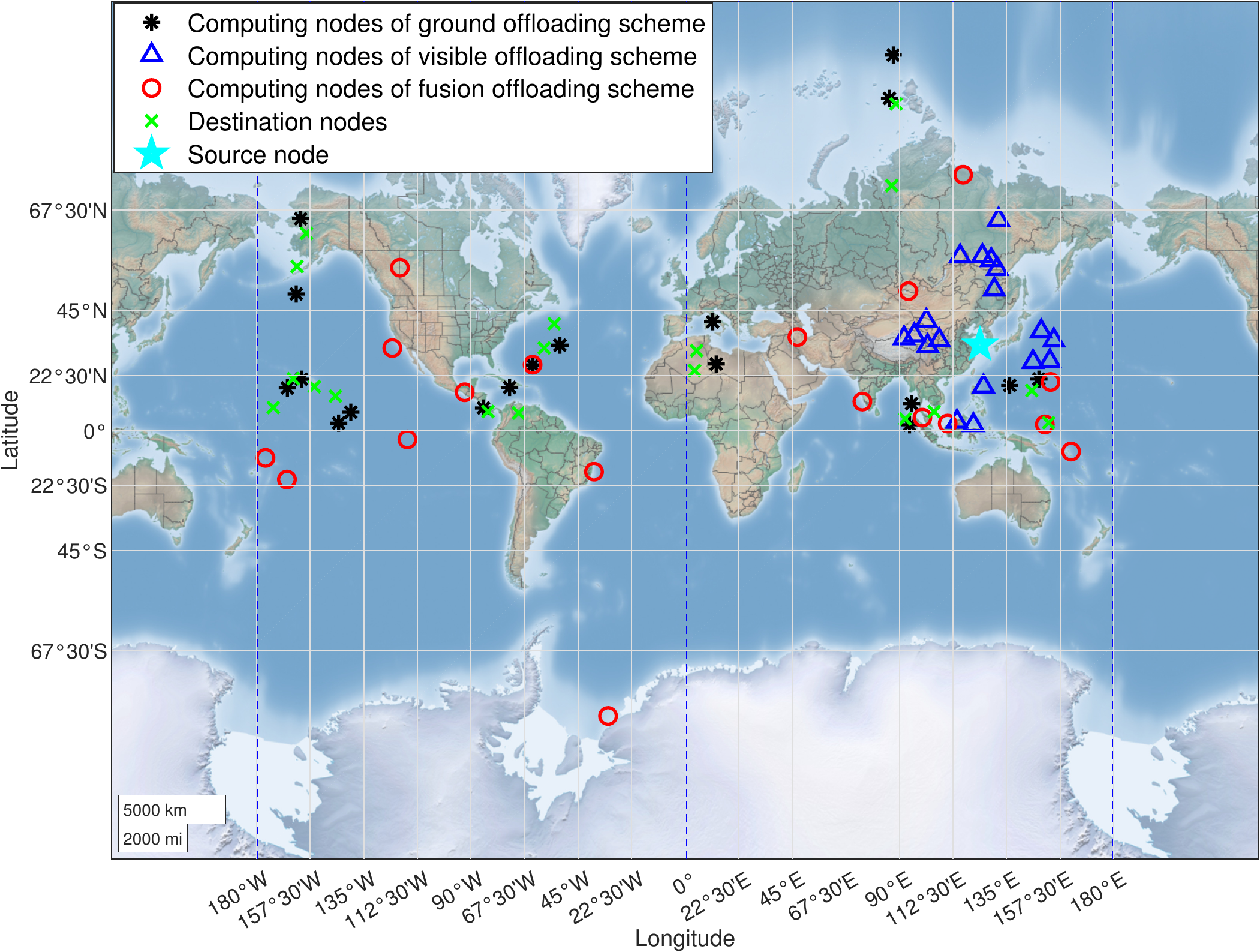}\\
	\begin{center}
		\caption{The offloading target distribution for tasks generated in a selected hot-spot (marked with a cyan star). Each dotted square represents a zone stated in Section~\ref{SystemModel} and corresponds to a node.}\label{Target_Distribution}
	\end{center}
\end{figure}

Inside Fig.~\ref{Target_Distribution}, blue triangles (indicates the target chosen by visible offloading) are located around the cyan star (i.e., the source of these tasks); therefore, only limited computing resources on visible satellites can be utilized, which constricts the performance of this scheme severely.
The black star shows the distribution of offloading targets of the ground offloading scheme. It can be seen that the black stars are located in the same zone as the destinations of tasks, which means these tasks are offloaded to the corresponding destination (i.e., ground servers) for computation.
The red circles show the distribution of offloading targets of the fusion offloading scheme. It can be concluded that the offloading targets of this scheme are distributed over the entire network.
More specifically, there are some tasks that are offloaded to zones other than areas where visible satellites and destinations are located. These zones correspond to non-hot-spot areas and usually have computing resources available. In this way, the excessive loads in hotspots can be relieved, as well as avoiding long delays caused by transmitting huge amounts of raw data to the ground.

\begin{mdframed}
	The proposed fusion offloading scheme can effectively reduce the overall delay. In particular, the proposed scheme decreases the weighted average delay up to 87.51\% and 18.70\% compared with the ground offloading scheme and the visible offloading scheme when $ \mathcal{L}=300 $.
\end{mdframed}


\subsection{Impacts from LEO Satellite Network Properties (RQ2)}
Computing capability can be a major factor for LEO-satellite-based offloading schemes. We investigate how weighted average delays of these offloading schemes vary with the computing capability of each LEO satellite.
The simulation covers computing capability from 50 GFLOPS to 500 GFLOPS. Its result is presented in Figure~\ref{Change_satellite_computing_capability}.

\begin{figure}[t]
	\vspace{-0.8em}
	\centering
	\setlength{\abovecaptionskip}{-0.cm}
	\setlength{\belowcaptionskip}{-0.cm}
	\includegraphics[width=0.5\textwidth]{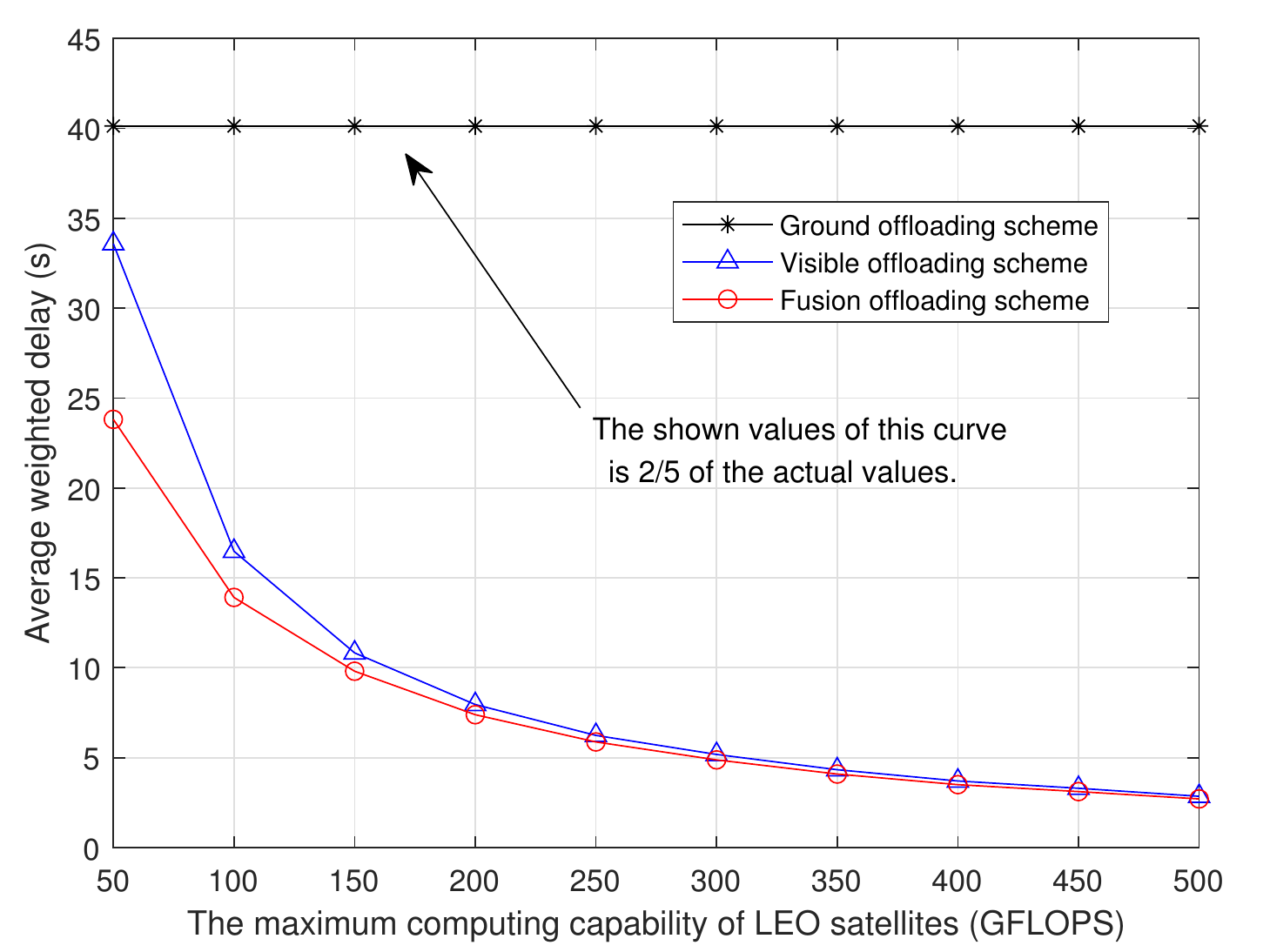}\\
	\begin{center}
		\caption{The weighted average delay versus the maximum computing capability of each LEO satellite.}\label{Change_satellite_computing_capability}
	\end{center}
\end{figure}

Fig.~\ref{Change_satellite_computing_capability} shows that the proposed fusion offloading scheme can achieve the lowest weighted average delay regardless of the change of the LEO satellite computing capability. 
Specifically, the curve of the ground offloading scheme remains constant; tasks are computed by ground servers rather than LEO satellites in this scheme. With the increase of the LEO satellite computing capability, the weighted average delays of the fusion offloading scheme and the visible offloading scheme decrease and the gap between them are getting smaller.
This is because with the increase of the LEO satellite computing capability, the time of each task occupying computing resources of LEO satellites is getting shorter, the probability of having available computing resources around is increasing, so the weighted average delays of both schemes are reduced. 
It is because the probability that the computation delay (including waiting delay) is smaller than the raw data transmission delay on the satellite network is increasing). Specifically, when the computing capability is strong enough, all tasks can be offloaded to visible satellites for computation, at which point the two curves overlap. \comment{Explain that you are the best.}
It can be seen from Fig.~\ref{Change_satellite_computing_capability} that the superiority of the proposed scheme is more obvious when the computing capability of satellites is relatively weak. This is because the bottleneck of the visible offloading scheme becomes more prominent (less available computing resources resulting in longer waiting delay) when the computing capability of satellites gets weaker. Therefore, the proposed scheme has significant superiority in decreasing delays, especially for networks with limited computing resources.

SGL transmission rate is also a factor affecting the performance. To investigate this factor, we change the rate from 0.2 Gbps to 10 Gbps and the remaining parameters follow Table~\ref{T1} and present the data in Fig~\ref{Change_SGL_data_transmission_rate}.

\begin{figure}[htbp]
	\centering
	\setlength{\abovecaptionskip}{-0.cm}
	\setlength{\belowcaptionskip}{-0.cm}
	\includegraphics[width=0.5\textwidth]{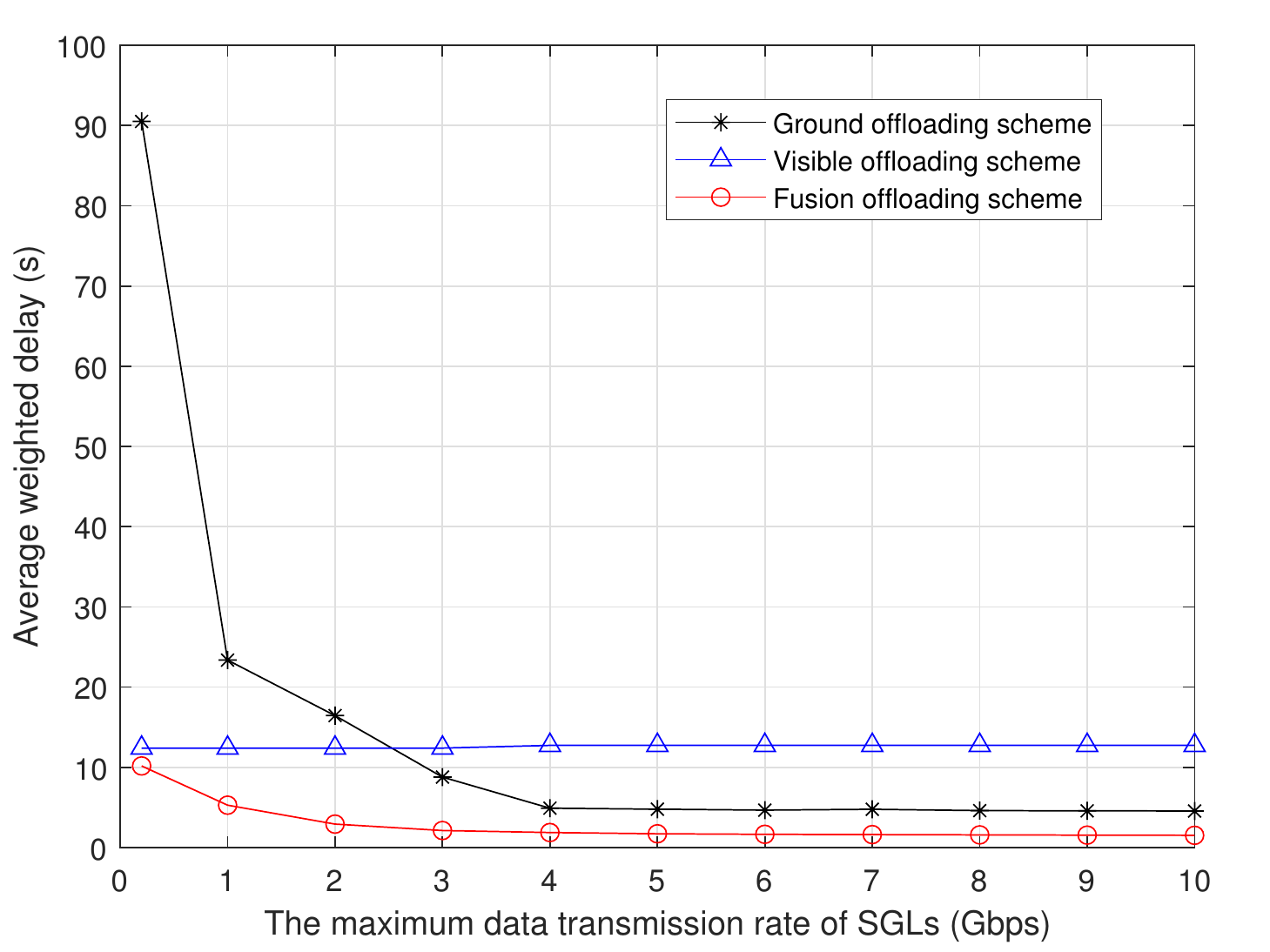}\\
	\begin{center}
		\caption{The weighted average delay versus the maximum data transmission rate of SGLs.}\label{Change_SGL_data_transmission_rate}
	\end{center}
\end{figure}

Fig~\ref{Change_SGL_data_transmission_rate} shows that the proposed fusion offloading scheme can achieve the lowest weighted average delay regardless of the change of data transmission rate of SGLs. 
Specifically, the weighted average delay of the visible offloading scheme remains horizontal as the rate of SGLs increases. It is because tasks are computed on visible satellites and only the computing results (usually only a few bites) are transmitted via SGLs.
Whereas, the weighted average delays of the fusion offloading scheme and the ground offloading scheme decrease as the rate of SGLs becomes larger. The proposed fusion offloading scheme performs better than the ground offloading scheme  because some tasks in the former are computed onboard which can save transmission delay. The difference of these two schemes decreases as the data transmission rate of SGLs increases. It is because the increase of SGL transmission rate lowers the satellite-to-ground transmission cost; as a result, more tasks are offloaded to the ground. \comment{Explain that our scheme is especially good on poor networks.}



\begin{mdframed}
	The proposed fusion offloading scheme can be applied to a wide range of networks. Its superior performance holds even for networks with low computing capabilities and low SGL data rates.
\end{mdframed}

\subsection{Impacts from Task Properties  (RQ3)}
Tasks with large computational requirements can be challenging for LEO-satellite-based offloading schemes.
To analyze the impact from a task's computational requirement, we simulate the computational requirement $ \widetilde{C} $ from 100 Giga floating-point operations (GFLO) to 1000 GFLO and present how the computational requirement $ \widetilde{C} $ of each subtask affects the task offloading performance in Fig~\ref{Change_task_computing_requirement}.

\begin{figure}[htbp]
	\centering
	\setlength{\abovecaptionskip}{-0.cm}
	\setlength{\belowcaptionskip}{-0.cm}
	\includegraphics[width=0.5\textwidth]{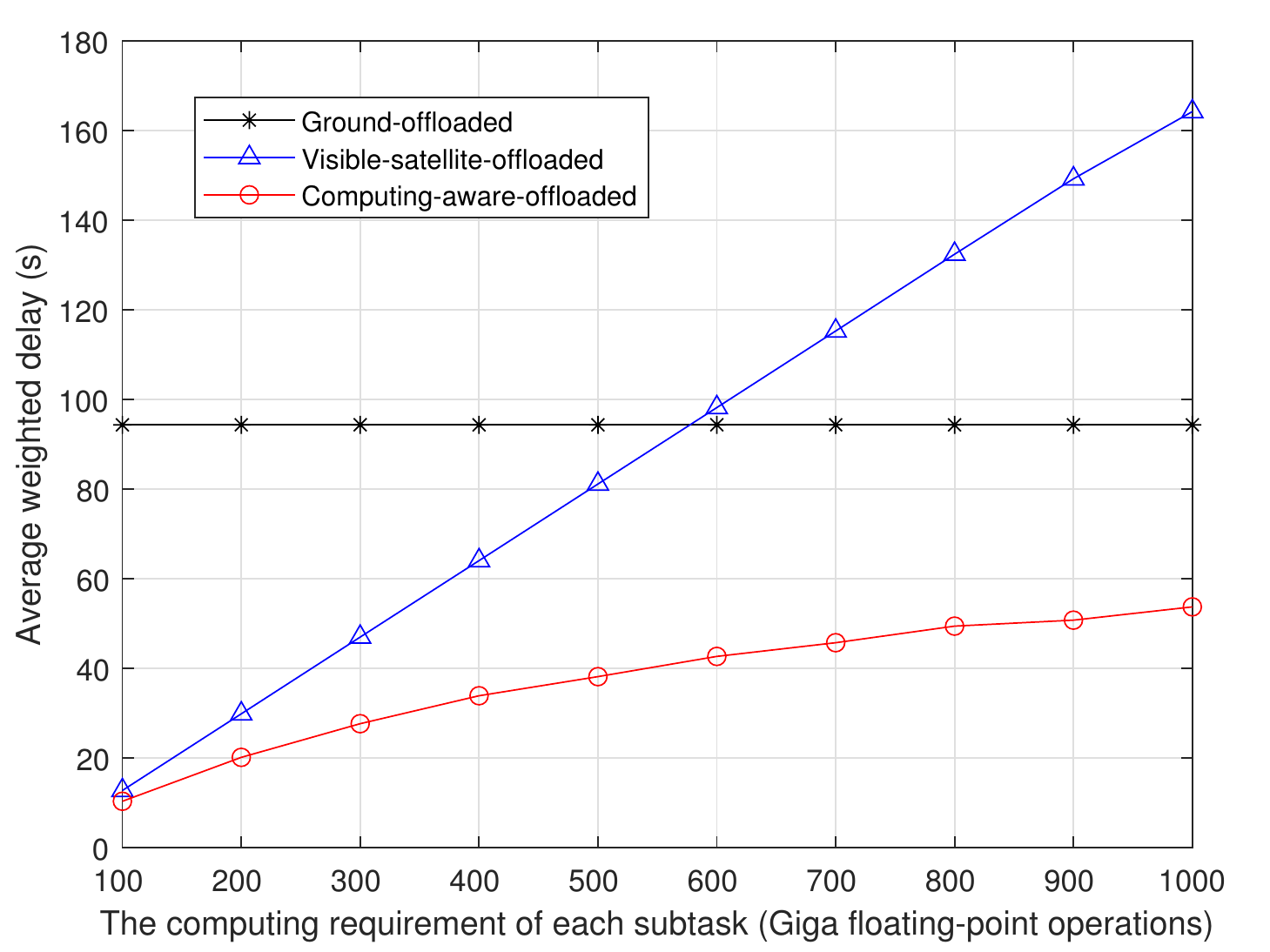}\\
	\begin{center}
		\caption{The weighted average delay versus the computational requirement $ \widetilde{C} $ of each subtask.}\label{Change_task_computing_requirement}
	\end{center}
\end{figure}

As shown in Fig~\ref{Change_task_computing_requirement}, the proposed fusion offloading scheme can achieve the lowest weighted average delay no matter how $ \widetilde{C} $ changes.
Specifically, the weighted average delay of the proposed fusion offloading scheme and the visible offloading scheme increases as the computational requirement $ \widetilde{C} $ of each subtask increases.
The root cause is that both schemes conduct onboard computing: the computation delay increases as the computational requirement increases when the satellite computing capability is fixed.
The growth rate of the visible offloading scheme is higher because this scheme can only utilize the computing resources on visible satellites; therefore, the waiting delay before performing the computation increases rapidly when the computational requirement increases.
On the contrary, The growth rate of the fusion offloading scheme is much lower. It is because this scheme can reduce the overall delay by offloading tasks to idle computing resources on invisible satellites and thus using free computing resources in non-hotspot areas. In this way, the excessive loads in hotspots can be released. In addition, since the fusion offloading scheme could offload tasks to ground servers for computation, the curve of the fusion offloading scheme approaches the curve of the ground offloading scheme, but never higher than the curve of the ground offloading scheme.
In addition, we assume that the computing capabilities of ground servers are strong enough in this paper; therefore, the computation delay is ignored when they are offloaded to ground servers for computing. Thus, the curve corresponding to the ground offloading scheme is horizontal.
It can be seen from Fig~\ref{Change_task_computing_requirement} that the superiority of the proposed scheme is getting more obvious when the computational requirement of tasks becomes higher. This is because the bottleneck of the visible offloading scheme becomes more prominent (higher computational requirements resulting in less computing resources available in the LEO satellite network) when the computational requirement of tasks becomes higher. Therefore, the proposed scheme has significant superiority in decreasing delays, especially for tasks with higher computational requirements.

Data volume could be another factor limiting the performance of the proposed scheme.
We investigate its impact by changing the data volume parameter $ \widetilde{N} $ from 1/5 GB to 1 GB and present the result in Fig~\ref{Change_subtask_data_volume}.

\begin{figure}[htbp]
	\centering
	\setlength{\abovecaptionskip}{-0.cm}
	\setlength{\belowcaptionskip}{-0.cm}
	\includegraphics[width=0.5\textwidth]{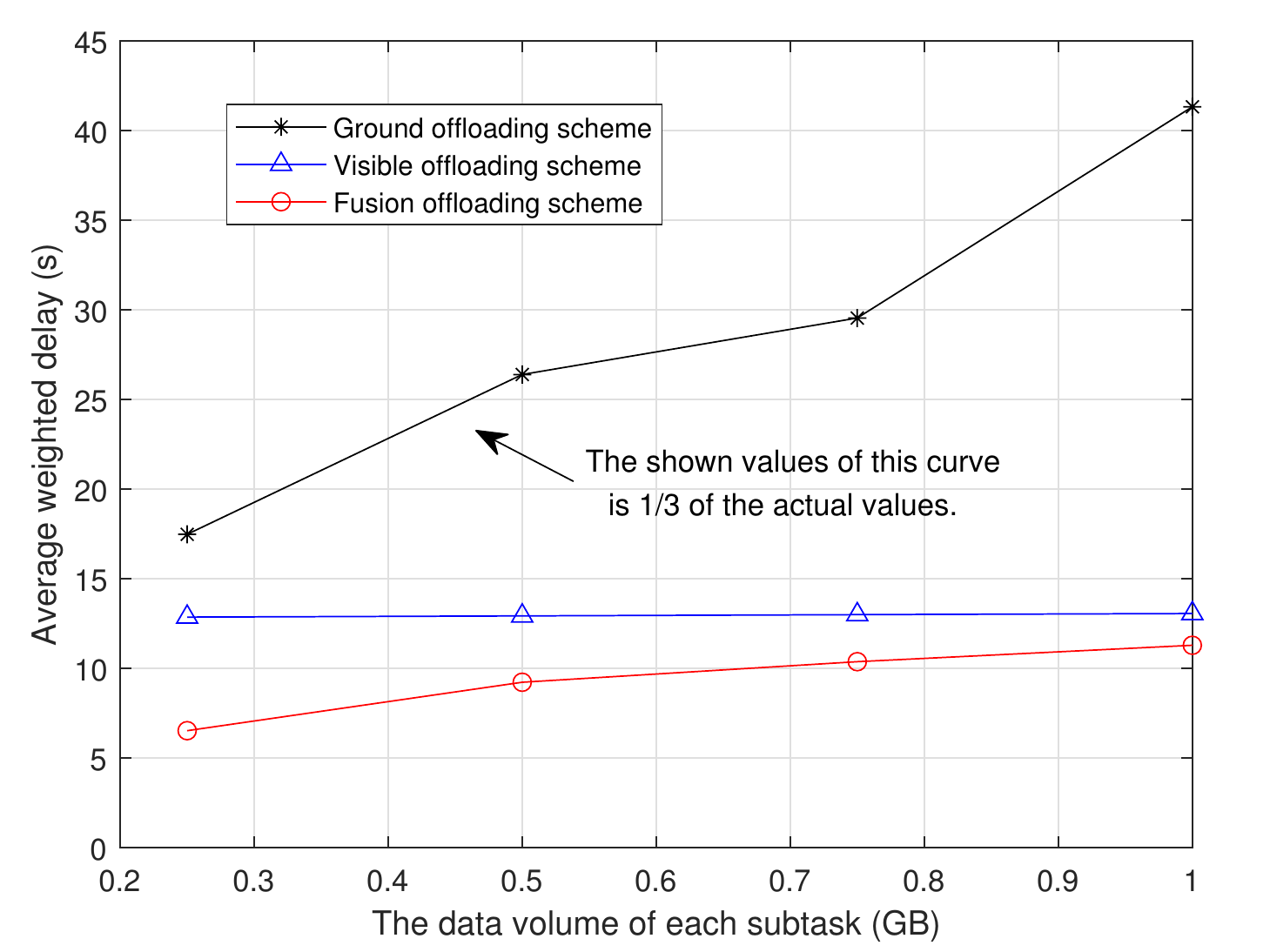}\\
	\begin{center}
		\caption{The weighted average delay versus the data volume $ \widetilde{N} $ of each subtask.}\label{Change_subtask_data_volume}
	\end{center}
\end{figure}

As shown in Figure~\ref{Change_subtask_data_volume},
the weighted average delays of the proposed fusion offloading scheme and the benchmark offloading schemes increase as the data volume $ \widetilde{N} $ of each subtask increase.
It is because all these schemes involve the transmission of raw data. The growth rate of the ground offloading scheme is the largest because raw data is transmitted in the whole offloading process.
On the contrary, the weighted average delay of the visible offloading scheme increases extremely slowly. It is because only one-hop transmission of raw data (i.e., from the source nodes to their visible LEO satellites) are contained in this scheme.
The growth rate of the fusion offloading scheme falls in between. In addition, the gap between the fusion offloading scheme and the visible scheme is getting smaller when $ \widetilde{N} $ increases.
In the fusion offloading scheme, tasks are offloaded to invisible satellites when the increased delay of transmitting raw data to invisible satellites (i.e., cost) is less than the decrease of the computation delay (i.e., gain). As $ \widetilde{N} $ keeps increasing, the cost gradually increases whereas the gain remains the same; therefore, more and more tasks are offloaded to visible satellites. \comment{Explain that our scheme is especially good on poor networks.}

\begin{mdframed}
	The proposed fusion offloading scheme consistently outperforms existing offloading schemes, even for tasks with relatively large computational requirement or data volume.
\end{mdframed}

\section{Conclusions}
This paper proposes a novel task offloading scheme for LEO satellite networks, which fuses the computation and transmission. In the proposed scheme, the transmission and computation are represented uniformly by introducing virtual edges. In addition, we proposed a metagraph-based transmission and computation fusion network construction method, which enables network-wide task offloading and reduces overall delay by jointly optimizing transmission and computation.

The superior performance of the proposed task offloading scheme over conventional schemes is demonstrated by both theoretical analyses and simulations. 
By allowing multi-hop offloading, the proposed scheme can effectively relieve the pressure of computing resource shortage caused by task overload in hotspots. 
By enabling onboard computing, the proposed scheme can overcome the limitation of excessive delay generated by transmitting large amounts of raw data over SGLs.
Simulation results show that the proposed offloading scheme can always obtain the lowest weighted average delay and the highest successful offloading rate no matter how the network properties and task properties change. 

The proposed task offloading scheme can make full use of network-wide computing resources, reduce delay under actual load distribution and balances the network load. The proposed scheme is important for enhancing the performance and facilitating the application of LEO satellite network task offloading. In the future, we will work on constructing a fusion network which can model the dynamics of the LEO satellite network directly rather than with the aid of virtual nodes, and the corresponding offloading strategies.
\bibliographystyle{IEEEtran}
\bibliography{IEEEfull,Reference}

\end{document}